\documentclass[headsepline,10pt,a4paper]{amsart}
\usepackage{amsmath,amssymb,amsfonts,amsthm,exscale,calc}
\usepackage[latin1]{inputenc}
\usepackage{latexsym,textcomp}
\usepackage{graphicx}
\usepackage{parskip}
\usepackage{floatflt}
\usepackage{picins}
\usepackage{dsfont}
\usepackage{hyperref}
\usepackage{vmargin}

\newtheorem{thm}{Theorem}

\newtheorem{lemma}{Lemma}
\newtheorem{prop}{Proposition}

\theoremstyle{definition}
\newtheorem*{Remark}{Remark}
\newtheorem*{definition}{Definition}
\newtheorem*{ex}{Examples}
 \newlength\headseptemp

\newcommand{\Hm}[1]{\leavevmode{\marginpar{\tiny%
$\hbox to 0mm{\hspace*{-0.5mm}$\leftarrow$\hss}%
\vcenter{\vrule depth 0.1mm height 0.1mm width \the\marginparwidth}%
\hbox to 0mm{\hss$\rightarrow$\hspace*{-0.5mm}}$\\\relax\raggedright
#1}}}

\newcommand{\Z}{{\mathbb Z}}
\newcommand{\R}{{\mathbb R}}
\newcommand{\C}{{\mathbb C}}
\newcommand{\h}{{\mathbb H}}
\newcommand{\N}{{\mathbb N}}
\newcommand{\EE}{{\mathbb E}}
\newcommand{\PP}{{\mathbb P}}

\newcommand{\T}{{\mathbb T}}

\newcommand{\A}{{\mathcal A}}

\newcommand{\Sp}{{\mathbb S}}

\newcommand{\Lp}{{T}}

\newcommand{\Vis}{\mathrm{Vis}}
\newcommand{\Leb}{\mathrm{Leb}}

\newcommand{\clos}{{\mathrm {clos}\,}}

\newcommand{\per}{{\mathrm {per}}}
\newcommand{\qand}{{\quad\mathrm {and}\quad}}
\newcommand{\qqand}{{\qquad\mathrm {and}\qquad}}

\newcommand{\ka}{{\kappa}}
\newcommand{\al}{{\alpha}}
\newcommand{\be}{{\beta}}
\newcommand{\de}{{\delta}}
\newcommand{\Gm}{{\Gamma}}
\newcommand{\gm}{{\gamma}}
\newcommand{\ph}{{\varphi}}
\newcommand{\lm}{{\lambda}}
\newcommand{\te}{{\theta}}
\newcommand{\eps}{{\varepsilon}}
\newcommand{\om}{{\omega}}
\newcommand{\Om}{{\Omega}}
\newcommand{\si}{{\sigma}}

\newcommand{\ap}[1]{\left( #1\right)}
\newcommand{\ab}[1]{\left( #1\right)}

\newcommand{\as}[1]{\langle #1\rangle}

\newcommand{\set}[1]{\left\{ #1\right\}}
\newcommand{\ov}[1]{\overline{#1}}
\newcommand{\ow}[1]{\tilde{#1}}
\newcommand{\oh}[1]{\widehat{#1}}

\newcommand{\mo}[1]{\left\vert #1\right\vert}

\let\Im\undefined
\let\Re\undefined
\DeclareMathOperator{\Im}{Im}
\DeclareMathOperator{\Re}{Re}

\begin{document}
\title[AC spectrum for random operators on trees of finite cone type]{Absolutely continuous spectrum for random operators on trees of finite cone type}
\author{Matthias Keller}
\author{Daniel Lenz}
\author{Simone Warzel}
\subjclass[2000]{}

\begin{abstract}We study the spectrum of random operators on a large class of trees. These trees have finitely many cone types and they can be constructed by a substitution rule. The random operators are perturbations of Laplace type operators either by random potentials or by random hopping terms, i.e., perturbations of the off-diagonal elements. We prove stability of arbitrary large parts of the absolutely continuous spectrum for sufficiently small but extensive disorder.
\end{abstract}

\maketitle
\section{Introduction and main results}
We study the stability of absolutely continuous spectrum of Laplace type operators on trees under random perturbations. The physical background is known as  Anderson localization \cite{An}. It deals with the question how the conductivity properties of  a quantum system change under the presence of disorder. Mathematically this translates into the study of the spectral measures of the corresponding operators, see the monographs \cite{CFKS,CL,Sto} for details and further reference.

While in one dimensional situations the spectrum turns immediately into pure point spectrum under the presence of a random potential (see \cite{CKM,GMP,KuS} and for one dimensional trees see also \cite{Br}), it is expected that some parts of the absolutely continuous spectrum are preserved from dimension three on. This is known as the extended states conjecture. However, it has only been  proven in the case of regular trees \cite{ASW,FHS2,Kl1,Kl2} and in strongly related models \cite{FHH,FHS3,FHS4,Hal,KlS}.

In this work we generalize the geometric setting from regular trees to a much larger class. These trees are often called trees of finite cone type or periodic trees and they are characterized by their underlying substitution type structure. In \cite{KLW}, we showed that Laplace type operators exhibit finitely many bands of purely absolutely continuous spectrum which is stable under certain small deterministic perturbations.  Here, we consider perturbations by small random potentials and hopping terms and prove stability of absolutely continuous spectrum. As regular trees are a special case, this work generalizes the main results of  \cite{ASW,FHS2,Kl1,Kl2}.

For the case of regular trees, there are three methods known to show stability of absolutely continuous spectrum under perturbations by small random potentials. The focus of all these methods lies on the analysis of the Green functions which  becomes a random variable under the presence of a random perturbations. (The Green functions are the diagonal matrix elements of the resolvent and their imaginary part converge weakly to the densities of the spectral measures in the absolutely continuous regime.)
Firstly, the method of Klein \cite{Kl1,Kl3}  uses  super symmetry and a Banach space version of the implicit function theorem to show that the moments of the Green functions are continuous for small disorder, both in the strength of the disorder and in the energy. Secondly, the method \cite{ASW} analyzes a Lyapunov exponent to show $L^1$-continuity of the expected value of the Green function at disorder zero.  Thirdly, Froese, Hasler and Spitzer \cite{FHS2} use a fixed point analysis and hyperbolic geometry to bound moments of the Green functions on a binary tree. Their method was later generalized to arbitrary regular trees \cite{Hal2}.

While our approach owes to \cite{FHS1,FHS2}, we develop an advanced scheme to deal with the  more complex geometric situation. This is also reflected in the results. For instance, all our estimates are explicit which could be  useful to obtain lower bounds on the magnitude of the disorder that still allows for absolutely continuous spectrum. Moreover, we prove a certain continuity for the moments of the Green functions in the strength of the disorder.

Recently,  a new method was developed which maps  the whole phase diagram and in particular the region of absolutely continuous spectrum using a moment-generating function of the Green function on a regular tree \cite{AW}. This new scheme covers the perturbative regime as a special case. However, at present it does not yield information about the purity of the absolutely continuous spectrum.

The present paper is based on the PhD thesis \cite{Kel}.


\subsection{The model}
Let us recall the definition of  the trees of finite cone type as it was given in \cite{Kel,KLW}.  Let a finite set $\A$ be given. We refer to its elements as \emph{labels}. Moreover, we need a map
\begin{align*}
    M:\A\times\A\to\N_0,\quad (k,l)\mapsto M_{k,l},
\end{align*}
which we call the \emph{substitution matrix}.
Let $\T$ be a rooted tree with root $o$ and vertex set $V$. We consider a labeling of the vertices be given which is a surjective map $V\to \A$. 
We say $\T$ is generated by the substitution matrix $M$ if for each vertex $x$ with label $k$ the number of forward neighbors of label $l$ is equal to $M_{k,l}$. Note that up to graph isomorphisms there are at most $\#\A$ different rooted trees generated by $M$. Each of them is completely determined by the label of the root (and of course  $M$).

We impose three more conditions on $M$:
\begin{itemize}
  \item [$\mathrm{(M0)}$] If $\A$ consists of only one element, then $M\geq 2$ \emph{(non one dimensional)}.
  \item [$\mathrm{(M1)}$] $M_{k,k}\geq 1$ for all $k\in\A$ \emph{(positive diagonal)}.
  \item [$\mathrm{(M2)}$] For all $k,l\in\A$ there is $n=n(k,l)\in\N$ such that $(M^n)_{k,l}\geq 1$ \emph{(irreducibility)}.
\end{itemize}
Note that in \cite{Kel,KLW} we assumed that the matrix is primitive instead of irreducible in $\mathrm{(M2)}$ (that is to ask that there is an $n$ such that the matrix elements $M^{n}$ are all positive). This was for convenience only since primitivity is implied by irreducibility provided that at least one diagonal entry is positive which is guaranteed by $\mathrm{(M1)}$.
The assumption $\mathrm{(M1)}$ means geometrically that each vertex has a forward neighbor of its own kind. On the other hand, $\mathrm{(M2)}$ means that vertices of every label are always found in the forward tree of every vertex.

There is a one-to-one relationship between trees of finite cone type, finite directed (not necessarily simple) graphs and trees generated by a substitution matrix. In particular, every tree of finite cone type is a directed cover of a finite directed graph and vice versa. On the other hand, every finite directed graph can be encoded by a substitution matrix and vice versa. Hence, given a finite directed graph every tree generated by the corresponding  substitution matrix is a directed cover of this graph which is therefore a tree of finite cone type.  For investigations of random walks  on such trees see \cite{Ly,NW}.

We study random perturbations of the operators $\Lp$ on $\ell^{2}(V)$ acting as
\begin{align*}
(\Lp\ph)(x)=\sum_{y\sim x}\ph(y) +v^{\per}(x)\ph(x),
\end{align*}
where $v^{\per}:V\to\R$ is a label invariant potential, i.e., the values agree on vertices of the same label.
In \cite{Kel,KLW}, we studied the spectral theory of these operators (in \cite{Kel} a more general class of operators and in \cite{KLW} only the adjacency matrix to avoid certain technicalities). It is shown there that the spectrum $\si(\Lp)$ of $\Lp$ is purely absolutely continuous and consists of finitely many intervals.


Let $(\Om,\PP)$ be a probability space and let
$$(v,\te):\Om\times V\to(-1,1) \times(-1,1), \quad (\om,x)\mapsto (v_{x}^{\om},\te_{x}^{\om}),$$
be a measurable function that satisfies the
following two assumptions:
\begin{itemize}
\item [$\mathrm{(P1)}$] For all $x,y\in V$ the random variables $(v_x,\te_{x})$ and $(v_y,\te_{y})$ are independent if the forward trees of $x$ and $y$ do not intersect.
\item [$\mathrm{(P2)}$] For all $x,y\in V$ that share the same label the restrictions of the random variables $(v,\te)$ to the isomorphic forward trees of $x$ and $y$ are identically distributed.
\end{itemize}
Here, we say that two random variables $X$, $Y$ defined on two isomorphic subgraphs $G_{X}$, $G_{Y}$ are identically distributed if for every label invariant graph isomorphism $\psi$ between $G_{X}$ and $G_{Y}$ the random variables $X$ and $Y\circ\psi$ are identically distributed. We will illustrate these conditions by some examples in Section~\ref{s:stability}.

The random variables $\te_{x}$, $x\in V$, induce  random variables $\te(x,y)$ on the edges of the tree, via $\te(x,\dot x)=\te(\dot x,x)=\te_{x}$, where $\dot x$ is the unique predecessor of $x$ with respect to the root. 

Given a coupling constant $\lm\geq0$, we consider the random operators $H^{\lm,\om}$ on $\ell^{2}(V)$ acting as
\begin{align*}
(H^{\lm,\om}\ph)(x)=
\sum_{y\sim x}\big(\ph(y)+\lm \te^{\om}(x,y)\ph(y)\big)+ v^{\per}(x)\ph(x) +\lm  v_{x}^{\om}\ph(x).
\end{align*}

For $\te\equiv0$, we get the Schr\"odinger operator $\Lp+\lm  v_{x}^{\om}$. On the other hand, for $v\equiv 0$ we get an operator with random hopping terms. This random hopping term model is also studied under the name first passage percolation. There, a  graph is equipped with random  edge weights which are interpreted as passage times and then transit times are studied, see  \cite{Kes} for a survey. Here, we focus on the spectral properties of the corresponding operators rather than transit times.



\subsection{Main results}
We study the model introduced in the previous section. Our main result states that arbitrary large parts of the absolutely continuous spectrum of $\Lp$ are stable under perturbations by sufficiently small random potentials.

\begin{thm} \label{main1}{There exists a finite set $\Sigma_{0}\subset\si(\Lp)$ such that for all compact $I\subset\si(\Lp)\setminus \Sigma_0$ there is $\lm$ such that
$H^{\lm,\om}$ has almost surely purely absolutely continuous spectrum in $I$ for all  $(v,\te)$ satisfying $\mathrm{(P1)}$ and $\mathrm{(P2)}$.}
\end{thm}

\begin{Remark}
(a) The theorem includes regular trees as a special case. Therefore, it generalizes the results of \cite{ASW,FHS2,Kl1} to a much more general class of trees. The potentials treated in \cite{ASW,FHS2} are allowed to be unbounded, (provided that certain bounds on the moments exist). For the sake of brevity, we restrict ourselves to the case of bounded potentials. Moreover,
in \cite{ASW} the potentials are allowed to be weakly correlated. However, this method does not yield purity of the absolutely continuous spectrum.

(b) We will actually prove similar result on trees where we replace $\mathrm{(M1)}$ by a weaker condition $\mathrm{(M1^*)}$, see Theorem~\ref{main3}.  For these trees we show that absolutely continuous spectrum is preserved in the subset of  $\si(\Lp)$, where we can guarantee positivity and continuity of the Green functions.

(c) In many cases, $\Sigma_{0}$ is explicitly given by the boundary points of $\si(\Lp)$ and possibly the point $0$, (for details see \cite{Kel,KLW}).

(d) The validity of the theorem does not depend on the particular choice of $\Lp$. In particular, $\Lp$ can be any nearest neighbor operator that is invariant with respect to the labeling of the tree as introduced in  \cite{Kel}.
\end{Remark}

The paper is structured as follows: In the next section we introduce the basic quantities of our analysis. We state a theorem that implies Theorem~\ref{main1}. Moreover, we discuss the strategy of proof. The proof relies on three crucial estimates which are then proven in Sections~\ref{s:TSE},~\ref{s:UC} and~\ref{s:VI}. In Section~\ref{s:proofs} we finally prove the theorem.


\section{Continuity of the Green function at low disorder}
\subsection{Basic properties of the Green function}

Let $\T$ be a rooted tree and let $H$ be a  nearest neighbor operator on $\ell^{2}(V)$, i.e., there is a symmetric function $t:V\times V\to\R$ such that $t(x,y)$ is non-zero if and only if $x$ and $y$ are adjacent  and $v:V\to\R$ such that $H$ acts as
\begin{align*}
(H\ph)(x)=\sum_{y\sim x}t(x,y) \ph(x)+v(x)\ph(x).
\end{align*}
We additionally assume that $H$ is bounded and self adjoint.
Let $\h$ denote the complex upper half plane, i.e., $\h=\{z\in\C\mid \Im z>0\}$. Let $\mu_{x}$ be the spectral measure of $H$ with respect to the characteristic function $\de_{x}$ of $x$. We let  the Green function  at a vertex $x$ be given by the Borel transform of $\mu_{x}$, i.e.,
\begin{align*}
G_{x}(z,H):=\int_{\si(H)}\frac{1}{t-z}d\mu_{x} =\as{\de_{x},(H-z)^{-1}\de_{x}}, \qquad z\in\h.
\end{align*}
It is well known that the Green function is an analytic map from $\h$ to $\h$.

A rooted tree has a natural ordering of the vertices with respect to their distance to the root in terms of the natural graph metric. We say a vertex lies in the $n$-sphere, if it is has distance $n$ to the root. Moreover, we denote the forward neighbors of a vertex $x$ by $S_{x}$ and the forward tree of $x$ by $\T_{x}$ with vertex set $V_{x}$.

It will be convenient to study the Green function of truncated operators $H_{\T_{x}}$, which are the restriction of $H$ to $\ell^{2}(V_x)$ considered as a subset of $\ell^{2}(V)$.
We denote the Green function of $H_{\T_x}$ at $x$ by
\begin{align*}
\Gm_{x}(z,H):=G_{x}(z,H_{\T_{x}}),\qquad z\in\h.
\end{align*}
We  call $\Gm_{x}$ the truncated Green functions, whenever a clear distinction from $G_{x}$ is needed.
It is well known that the truncated Green functions satisfy the following recursion relation
\begin{align}\label{e:rec}\tag{$\clubsuit$}
-\frac{1}{\Gm_{x}(z,H)}=z-v(x)+\sum_{y\in S_{x}}|t(x,y)|^2\Gm_{y}(z,H),\quad z\in\h.
\end{align} This formula is a direct consequence of a twofold application of the resolvent formula, for references see \cite{ASW,Kel,KLW,Kl1}.

In the following let  $\T$ be a tree of finite cone type. Clearly, $\Gm_{x}(\cdot,\Lp)=\Gm_{y}(\cdot,\Lp)$, whenever $x$ and $y$  have the same label. In this case, \eqref{e:rec} becomes a finite system of equations.


We next introduce a subset of $\si(\Lp)$ on which we prove stability of absolutely continuous spectrum.
\begin{definition}
Let $\mathcal{U}$ be the system of all  open sets $U\subset \R$ such that for each $x\in V$
 the function $\h\to\h,\; z\mapsto \Gm_{x}(z,\Lp)$ can uniquely be extended to a  continuous
 function from  $\h\cup U$ to $ \h$ and set
\begin{align*}
\Sigma:=\bigcup_{U\in \mathcal{U}} U.
\end{align*}
\end{definition}
As a consequence,
for all $E\in\Sigma$  and $x\in  V$ the limits  $\Gm_{x}(E,\Lp)=\lim_{\eta\to0}\Gm_{x}(E+i\eta,\Lp)$
exist, are continuous functions in $E$ and $\Im \Gm_{x}(E,\Lp)>0$. Since the measures $\Im G_{x}(E+i\eta,\Lp)dE$ converges weakly to the spectral measure $\mu_{x}$, we have $\Sigma\subset \si_{\mathrm{ac}}(\Lp_x)$ for all $x\in V$ and, therefore, $\Sigma\subset \si_{\mathrm{ac}}(\Lp)$.

Indeed, for the trees satisfying $\mathrm{(M0)}$, $\mathrm{(M1)}$, $\mathrm{(M2)}$, we know the following, see  \cite[Theorem~3.1]{Kel}, \cite[Theorem~6]{KLW}.

\begin{prop}\label{p:T} Let $\T$ be a tree generated by a substitution matrix satisfying $\mathrm{(M0)}$, $\mathrm{(M1)}$, $\mathrm{(M2)}$. Then, the set $\Sigma$ consists of finitely many intervals and
$$\clos\Sigma=\si_{\mathrm{ac}}(\Lp)=\si(\Lp).$$
\end{prop}


\subsection{A stability result}\label{s:stability}

We will  prove stability of absolutely continuous spectrum inside of $\Sigma$ for a more general class of trees than introduced in the previous section. We replace $(M1)$ by the following weaker assumption:
\begin{itemize}
\item [$\mathrm{(M1^*)}$] For each $k\in\A$ there is $k'\in\A$ with $M_{k,k'}\geq1$ such that for each $l\in\A$ with $M_{k,l}\geq1$ we have $M_{k',l}\geq1$.
\end{itemize}
This has the following consequence for the trees generated by a substitution matrix satisfying this assumption: For each vertex $x$ there is a vertex $x'$ in $S_x$ such that each label found in $S_{x}$ can  also be found in $S_{x'}$. We will then say that the forward neighbor $x'$  of $x$ is chosen with respect to $\mathrm{(M1^*)}$.

As we assume that the tree is not one-dimensional, the assumption $\mathrm{(M1*)}$ implies that each vertex has at least two forward neighbors. Clearly,  $\mathrm{(M1)}$ implies $\mathrm{(M1^*)}$.


For a probability space $(\Om,\PP)$ and an integrable function $f$ on $\Om$ we denote its expected value by
\begin{align*}
    \EE(f)=\int_{\Om }f(\om)d\PP(\om).
\end{align*}
Theorem~\ref{main1} is a consequence of Proposition~\ref{p:T} and the following theorem. A proof will be given in Section~\ref{s:proofs}.

\begin{thm} \label{main3}{Let $\T$ be a rooted tree such that the forward trees of all vertices from a certain sphere on are generated by  substitution matrices that satisfy $\mathrm{(M0)}$, $\mathrm{(M1^*)}$ and $\mathrm{(M2)}$.
For all compact $I\subset\Sigma$ and $p>1$, there is $\lm_0=\lm_0(I)>0$ and $C_{x}:[0,\lm_0)\to[0,\infty)$ for each $x\in V$ with $\lim_{\lm\to0}C_{x}(\lm)=0$ such that
\begin{align*}
\sup_{z\in I+i(0,1]} \EE\ab{|G_{x}(z,H^{\lm})-G_{x}(z,\Lp)|^{p}} \leq C_{x}(\lm)
\end{align*}
for all $\lm\in[0,\lm_0)$ and all $(v,\te)$ satisfying $\mathrm{(P1)}$ and $\mathrm{(P2)}$.
In particular,
$H^{\lm,\om}$ has almost surely purely absolutely continuous spectrum in $I$ for all $\lm\in[0,\lm_0)$.}
\end{thm}

\begin{Remark}
(a)  The statement about convergence of the Green functions in Theorem~\ref{main3} is stronger than the results obtained in \cite{ASW} and \cite{FHS2} for regular trees. In \cite[Theorem~1.4]{Kl1} an even stronger statement is found for regular trees.

(b) Many examples of directed or universal covers of finite graphs are covered by the assumptions of Theorem~\ref{main3}. For example, the universal cover of every finite (not necessary simple) graph, where the minimal vertex degree is at least three  and where every vertex has loop, can be seen to satisfy the assumptions.
\end{Remark}

Let us close this section by giving some examples for random operators that satisfy $\mathrm{(P1)}$ and $\mathrm{(P2)}$.

\begin{ex} Let $\T$ be a tree that satisfies the assumption of Theorem~\ref{main3}.

(a) Let $v_{x}$, $x\in V$, be independent  random variables. Then $(v,0)$ satisfies $\mathrm{(P1)}$. If additionally $v_{x}$ and $v_{y}$ are  identically distributed whenever $x$ and $y$ have the same label, then, also $\mathrm{(P2)}$ is satisfied.

(b) Let $w_{x}$,  $x\in V$, be independent and  identically distributed random variables. Moreover, let $k$ be the maximal number of forward neighbors of vertices in $\T$. Then, the potential given by
\begin{align*}
    v_{x}=\sum_{y\in\T_{x}}\frac{1}{k^{d(x,y)+1}}w_{y},
\end{align*}
where $d(\cdot,\cdot)$ is the natural graph distance, together with $\te\equiv0$ satisfies $\mathrm{(P1)}$ and $\mathrm{(P2)}$.

(c) We assign an additional distinct label to the root and let $v^{\per}=-\deg$, where $\deg$ is the vertex degree.
For independent and  identically distributed random variables $\te_{x}$,  $x\in V$,   and $v_{x}=\te_{x}+\sum_{y\in S_{x}}\te_{y} =\sum_{y\sim x} \te(x,y)$, the pair $(v,\te)$ satisfies $\mathrm{(P1)}$ and $\mathrm{(P2)}$. Then, $H^{\lm,\om}$ is a weighted Laplacian with random edge weight, i.e.,
\begin{align*}
(H^{\lm,\om}\ph)(x)=\sum_{y\sim x} (1+\lm\te^{\om}(x,y))(\ph(y)-\ph(x)).
\end{align*}
\end{ex}


\subsection{Strategy of the proof}

The overall strategy of the proof of Theorem~\ref{main3} is inspired by \cite{FHS2}, however, the actual steps are very  different. This is on the one hand due to the fact that our geometric situation is essentially more complicated. On the other hand, our aims are somewhat higher since we show continuity of the moments of the Green function in the coupling constant.

In order to control the Green function of the perturbed operator, we measure its distance to the unperturbed Green function by means of a hyperbolic semi-metric (i.e., a symmetric and positive definite function, which not necessarily satisfies the triangle inequality).
It is given by a function $\gm:\h\times\h\to[0,\infty)$
\begin{align*}
    \gm(g,h)=\frac{|g-h|^{2}}{\Im g\Im h}.
\end{align*}
The function $\gm$ is related to the standard hyperbolic metric $d_{\h}$ of the upper half plane via $d_{\h}(g,h)=\cosh^{-1}(\tfrac{1}{2}\gm(g,h)+1)$.
This approach was introduced in \cite{FHS1,FHS2} and adapted  in \cite{Kel,KLW}. For more details on $\gm$ see \cite[Section~2.3]{Kel}.

Let us sketch the key ideas for the proof of Theorem~\ref{main3}. For $\lm\geq0$, $z\in\h$, we look at the random variable $\gm_{x}$  given by
\begin{align*}
{\gm_{x}^{\om}}={\gm{\ap{\Gm_x(z,H^{\lm,\om}),\Gm_x(z,\Lp)}}},\quad \om\in\Om.
\end{align*}
We use the recursion relation \eqref{e:rec} to expand this random variable at the root vertex $x=o$ in terms of its values on the next two spheres: Letting $S_{o,o'}=S_{o'}\cup S_{o}\setminus\{o'\}$, for $o'$ chosen with respect to $\mathrm{(M1^*)}$, we find non random \emph{weights} $p_{x}$ satisfying $\sum_{x\in S_{o,o'}} p_{x}=1$ and random \emph{contraction quantities } $c_{x}\leq1$ such that
\begin{align*}
\gm_{o}\leq (1+ C)\sum_{x\in S_{o,o'}}p_{x}c_{x}\gm_{x}+C
\end{align*}
with non random constants $C$ satisfying $C\to0$ as $\lm\to0$. This formula shows that the distance of the Green functions $\gm_{o}$ at vertex the $o$ can be estimated by a convex combination of the distances $\gm_x$ in $S_{o,o'}$ which are multiplied by the contraction quantities $c_x$.
We refer to this as a \emph{two step expansion estimate} since it is proven by first expanding $o$ in terms of $S_{o}$ and secondly expanding $o'$ in terms of $S_{o'}$ which then yields the sum over $S_{o,o'}$. For the details see Section~\ref{s:TSE}.  We will use the fact that  the random variables $\Gm_{x}^{\lm}:=( \Gm_x(z, H^{\lm}))$  are identically distributed for all $x\in V$ that carry the same label. Moreover, for all $x$ and $y$ with non intersecting forward trees they are independent. Hence, we can interchange the random variables $\Gm_{x}^{\lm}$ and $\Gm_{y}^{\lm}$ in the expected value whenever $x$ and $y$ have the same label and their forward trees do not intersect. We introduce the set $\Pi$ of  permutations of $S_{o,o'}$ that leave the label  of a vertex invariant. In Section~\ref{s:UC}, we show that an averaged contraction coefficient $\ka_{o}^{(p)}$, $p>1$, given by
\begin{align*}
\ka_{o}^{(p)}=\frac{\frac{1}{|\Pi|}\sum_{\pi\in\Pi} \ab{\sum_{x\in S_{o,o'}} p_{x}(c_{x}\circ\pi)(\gm_{x}\circ\pi)}^{p}} {\frac{1}{|\Pi|}\sum_{\pi\in\Pi} \sum_{x\in S_{o,o'}} p_{x}(\gm_{x}\circ\pi)^{p}}
\end{align*}
satisfies
$$\ka_{o}^{(p)}\leq 1-\de_0,$$
whenever some $\gm_{x}$ is large and $\lm$ is sufficiently small. The constant $\de_0>0$ does not dependent  on $z$ and $\lm$ (as long as $\Re z\in I$ for some compact  $I\subset\Sigma$ and $\lm$ is sufficiently small). We refer to this as a \emph{uniform contraction estimate} and it is found in Proposition~\ref{p:ka}, Section~\ref{s:UC}. This terminology comes from the fact that \eqref{e:rec} can be expressed as a map from the truncated Green functions of $S_{o}$ to the one at $o$ and the one of $S_{o'}$ to the one at $o'$. The composition as a map from $\h^{S_{o,o'}}$ to $\h$ combined with the averaging over the permutations in $\Pi$ can be then thought to be an contraction. 

We denote by $o(j)$ the root of the tree with root label $j\in\A$ and we consider the vector
 $(\gm_{o(j)})_{j\in\A}$. The $p_x$'s in the
two step expansion estimate are actually functions of  the unperturbed Green functions and depend therefore continuously on $z$ in $\h\cup\Sigma$. They give rise to a stochastic matrix $P:{\A\times\A}\to[0,\infty)$ and we will show, using the two step expansion and the uniform contraction estimate, that for  sufficiently small $\lm$ the  \emph{vector inequality}
\begin{align*}
\EE{\big((\gm_{o(j)}^{p})\big)}
&\leq (1-\de)P\EE{\big((\gm_{o(j)}^{p})\big)}+(C_j)
\end{align*}
holds,
with $\de>0$ independent of $z$ and $\lm$ and $C_j\to0$ as $\lm\to0$, $j\in\A$. Of course, the inequality is understood componentwise. It will be proven in Section~\ref{s:VI}.

Now, by the Perron-Frobenius theorem there is a positive left eigenvector $u$ of $P$ such that $P^{\top} u=u$, which also depends continuously on $z$. Hence,
\begin{align*}
\as{u,\EE\big((\gm_{o(j)}^{p})\big)} \leq(1-\de)\as{u,\EE\big((\gm_{o(j)}^{p})\big)}+C,
\end{align*}
where $\as{\cdot,\cdot}$ is the standard scalar product in $\R^{\A}$.
This yields immediately that the $j$-th component of the vector $\EE((\gm_{o(j)}^{p}))$ is smaller than $c/(u_j\de)$.

From this inequality we derive bounds on the moments of the Euclidean distance of the two Green functions as stated in the first statement of Theorem~\ref{main3}. The second statement of Theorem~\ref{main3} about preservation of the absolutely continuous spectrum then follows  from a variation of the limiting absorption principle.



\section{The two step expansion estimate}\label{s:TSE}
The aim of this section is to expand the $\gm$-distance of the  perturbed and the unperturbed Green function. The expansion yields a convex combination of the distances on the forward spheres which are multiplied by contraction quantities.

The proof of the estimate involves two steps. Firstly, we prove an inequality to deal with small linear perturbation of one argument in $\gm$. Secondly, we prove a one step expansion estimate. Finally, we come to the two step expansion estimate, Proposition~\ref{p:expansion}.


\subsection{Linear perturbation estimate}

\begin{lemma}\label{l:ti}
Let $a,b\in(-1,1)$, $h\in\h$. Then for all $g\in\h$ and $\lm\in[0,1]$
\begin{align*}
\gm((1+\lm a)g+\lm b,h)\leq(1+c_0(\lm,h))\gm(g,h)+c_{0}(\lm,h)
\end{align*}
where $c_0(\lm,h)=-1+(1+\lm |a|)(1+ 2\lm\ab{2|a||h|+ |b|}/{\Im h})^{2}$. In particular, the function $c_0$ is independent of $g$ and $\lim_{\lm\to0}c_0(\lm,h)=0$ for fixed $h$.
\end{lemma}
\begin{proof}
Denote $f_{\lm}(g)=(1+\lm a)g+\lm b$.
We  distinguish two cases: If $|g-h|\geq \Im h/2$, then,  using $|g|\leq|g-h|+|h|$,  we get
$$(1+|a|\lm)\gm(f_{\lm}(g),h)\leq \ab{1+\lm\frac{|a||g|+ |b|}{|g-h|}}^2 \gm(g,h)\leq \ab{1+\lm{\frac{3|a||h|+2|b|}{\Im h}}}^2 \gm(g,h).$$
If, on the other hand, $|g-h|\leq \Im h/2$, then $\Im g\geq\Im h/2$. We obtain
\begin{align*}
(1+|a|\lm)\gm(f_{\lm}(g),h)&\leq \gm(g,h)+\frac{2\lm(|a||g|+|b|)|g-h|+\lm^{2}(|a||g|+|b|)^2}{\Im g\Im h}\\
&\leq \gm(g,h)-1+
\ab{1+\lm\frac{4|a||h|+2|b|}{\Im h}}^{2}.
\end{align*}
The statement now follows from the definition of $c_0$.
\end{proof}

\subsection{A one step expansion estimate}\label{ss:OSE}

We consider a tree  $\T$ generated by a substitution matrix that satisfies $\mathrm{(M0)}$ and $\mathrm{(M1^*)}$.
For a vertex $o$ let $o'$ be a forward neighbor of $o$ chosen  with respect to $\mathrm{(M1^*)}$. (We may think of $o$ as the root of $\T$.) Let
$$S_{o,o'}:=S_{o'}\cup S_{o}\setminus \{o'\}.$$
For a subset $W\subseteq V$ and a vector $g\in \h^{V}$ we write $g_{W}$ for the  restriction of $g$ to $W$.
Next, we define functions of  vectors in $\h^{S_{o,o'}}$. These vectors are usually denoted by $g$ and they will later play the role of $\Gm_{S_{o,o'}}(z,H^{\lm,\om})$.  Having the recursion relation of the truncated Green functions \eqref{e:rec} in mind, we define for $g\in\h^{S_{o,o'}}$, $z\in\h$, $w,w',\vartheta\in(-1,1)$,
\begin{align*}
g_{o'}:=g_{o'}(z,w',g)&:=-\frac{1}{z-v^{\per}(o')-w'+\sum_{x\in S_{o'}}g_{x}},\\
g_{o'}:=g_{o'}(z,w,w',\vartheta,g)&:=-\frac{1}{z-v^{\per}(o)-w +(1+\vartheta)g_{o'}(z,w',g)+\sum_{x\in S_{o}\setminus\{o'\}}g_{x}}.
\end{align*}
We will often keep the dependence of $g_o$ and $g_{o'}$ on $z$, $w$, $w'$, $\vartheta$ implicit. Putting $\ow \Gm_{x}^{\lm,\om} :=(1+\lm \te_{x}^{\om})\Gm_x(z,H^{\lm,\om})$, $x\in{S_{o,o'}}$, we see that
\begin{align*}
\Gm_{o'}(z,H^{\lm,\om})&=g_{o'}(z,\lm v_{o'}^{\om},\ow\Gm_{S_{o,o'}}^{\lm,\om})\qand
\Gm_{o'}(z,H^{\lm,\om})=g_{o}(z,\lm v_{o}^{\om}, \lm v_{o'}^{\om},\lm \te_{o'}^{\om}, \ow\Gm_{S_{o,o'}}^{\lm,\om}).
\end{align*}
For a vertex $x\in\{o\}\cup S_{o}\cup S_{o'}$,
we define $\gm_{x}:\h^{S_{o,o'}}\to[0,\infty)$ by
\begin{align*}
    \gm_{x}(g):=\gm(g_{x},\Gm_{x}(z,\Lp)), \qquad g\in \h^{S_{o,o'}},
\end{align*}
where for $x=o$ and $x=o'$ we use the definition of $g_{o}$ and $g_{o'}$ above. We write $\Gm_{x}=\Gm_{x}(z,\Lp)$ for short and for $x,y\in S_{i}$, $i\in\{o,o'\}$, we let the functions $q_{y}, Q_{x,y}, \cos\al_{x,y}:\h^{S_{o,o'}}\to\R$ be given by
\begin{align*}
q_{y}(g)&:=\frac{\Im g_{y}}{\sum_{u\in S_{i}}\Im g_{u}},\\
Q_{x,y}(g)&:=\frac{\sqrt{\Im g_{x}\Im g_{y}\Im \Gm_{x}\Im \Gm_{y}\gm_{x}(g)\gm_{y}(g)}}{\frac{1}{2}(\Im g_{x}\Im \Gm_{y}\gm_{y}(g)+\Im g_{y}\Im \Gm_{x}\gm_{x}(g))},\\
\cos\al_{x,y}(g)&:=\cos\arg(g_{x}-\Gm_{x})\ov{(g_{y}-\Gm_{y})},
\end{align*}
provided that $g_{x}\neq \Gm_{x}$, $g_{y}\neq \Gm_{y}$. Otherwise, we put $Q_{x,y}(g)=\cos\al_{x,y}(g)=0$. The functions $\gm_{x}$, $Q_{x,y}$, $\al_{x,y}$ depend all on the unperturbed Green function and thus on $z$. We suppress this dependance to ease notation.

The next lemma is the one step expansion formula. A similar statement can be found in \cite[Lemma 5]{KLW}.

\begin{lemma}\label{l:OSE} (One step expansion estimate)
Let $I\subset \Sigma$ be  compact. Then there exist $c_0:[0,\infty)\to[0,\infty)$ with $\lim_{\lm\to0}c_0(\lm)=0$ such that for all $g\in \h^{S_{o,o'}}$, $z\in I+i[0,1]$, $\lm\in[0,1]$, $w,w',\vartheta\in(-\lm,\lm)$, $i\in\{o,o'\}$,
\begin{align*}
{\gm\ap{g_{i},h_{i}}}
&\leq(1+c_0(\lm))\sum_{x\in S_{i}} \frac{\Im \Gm_{x}(z,\Lp)}{\sum_{u\in S_{i}}\Im \Gm_{u}(z,\Lp)} \ab{\sum\limits_{y\in S_i} q_{y}(g)Q_{x,y}(g)\cos\al_{x,y}(g)}\gm_{x}(g) +c_0(\lm).
\end{align*}
\end{lemma}
\begin{Remark}
The weights $q_{y}$ add up to one, i.e., we have $\sum_{y\in S_{i}}q_{y}=1$. Moreover, the quantity $Q_{x,y}$ is a quotient of a geometric and an arithmetic mean and therefore takes values between $0$ and $1$. As also $-1\leq\cos\al_{x,y}\leq1$, we conclude that  $-1\leq\sum q_{y}Q_{x,y}\cos\al_{x,y}\leq 1$.
\end{Remark}

\begin{proof}
By a direct calculation it can be seen that for $\xi,\zeta\in\h$ and $z\in \h$
\begin{align*}
\gm\Big({-\frac{1}{z+\xi},-\frac{1}{z+\zeta}}\Big) =\gm(z+\xi,z+\zeta) =\gm(i\Im z+\xi,i\Im z+\zeta)\leq\gm(\xi,\zeta).
\end{align*}
Moreover, given $\xi_x,\zeta_x\in \h$, $x\in S_{o}\cup S_{o'}$, we calculate directly
\begin{align*}
\Big|{\sum_{x\in S_{i}}\xi_x-\sum_{x\in S_{i}}\zeta_x}\Big|^{2}&=\sum_{x,y\in S_{i}}\cos\al_{x,y}|\xi_{x}-\zeta_{x}||\xi_{y}-\zeta_{y}|\\
&=\sum_{x,y\in S_{i}}\cos\al_{x,y}\ab{\Im\xi_{x}\Im\zeta_{x}\Im\xi_{y}\Im\zeta_{y} \gm(\xi_{x},\zeta_{x})\gm(\xi_{y},\zeta_{y})}^{\frac{1}{2}}\\
&=\sum_{x\in S_{i}}\Im \zeta_{x} \Big({\sum_{y\in S_{i}} \Im \xi_{y}\cos \al_{x,y}Q_{x,y}}\Big)\gm(\xi_{x},\zeta_{x})
\end{align*}
for $i\in\{o,o'\}$.
Hence, combining these two facts  and Lemma~\ref{l:ti} we get the statement. The constant $c_0$ depends only on $\lm$ and the unperturbed truncated Green functions.
As these are continuous functions  from $\Sigma\cup\h$ to $\h$, their imaginary parts are uniformly larger than zero and their absolute value is uniformly bounded on $I+i[0,1]$. Hence, $c_{0}$  depends only on $\lm $ and $I$.
\end{proof}


\subsection{Statement and proof of the two step expansion estimate}

Based on the one step expansion estimate we make the following definitions: For $z\in I+i[0,1]$ and $w\in\R$, we define the \emph{contraction quantities}  $c_{x}:\h^{S_{o,o'}}\to[-1,1]$ for  $x\in S_{o}$ by
\begin{align*}
c_{x}(g)&:=\sum\limits_{y\in S_o} q_{y}(g)Q_{x,y}(g)\cos\al_{x,y}(g)
\end{align*}
and for $x\in S_{o'}$ by
\begin{align*}
c_{x}(g)&:=\Big({\sum\limits_{y\in S_{o}} q_{y}(g)Q_{x,y}(g)\cos\al_{x,y}(g)}\Big)\Big( {\sum\limits_{y\in S_{o'}} q_{y}(g)Q_{o',y}(g)\cos\al_{o',y}(g)}\Big).
\end{align*}
By the remark after Lemma~\ref{l:OSE} the contraction quantities $c_{x}$ take values in $[-1,1]$.
The parameters $z$ and $w$ enter via $g_{o'}$ and $z$ enters also via $\Gm_{x}(z,\Lp)$ into $Q$ and $\al$. We define the  \emph{weights} as functions of the unperturbed Green functions for $x\in S_{o}$
\begin{align*}
p_{x}:=
q_{x}(\Gm_{S_{o,o'}}(z,\Lp)) =\frac{\Im \Gm_{x}(z,\Lp)}{\sum_{y\in S_{o}}\Im \Gm_{y}(z,\Lp)}
\end{align*}
and for $x\in S_{o'}$
\begin{align*}
p_{x}:=q_{o'}(\Gm_{S_{o,o'}}(z,\Lp))q_{x}(\Gm_{S_{o,o'}}(z,\Lp))=
\frac{\Im \Gm_{o'}(z,\Lp)\Im \Gm_{x}(z,\Lp)}{\big(\sum_{y\in S_{o}}\Im \Gm_{y}(z,\Lp)\big)^{2}}.
\end{align*}
Clearly, $\sum_{x\in S_{o,o'}}p_{x}=1$.

The following proposition is the two step expansion.
\medskip

\begin{prop}\label{p:expansion}(Two step expansion estimate)
Let $I\subset \Sigma$ be compact. Then there exist $c:[0,\infty)\to[0,\infty)$ with $\lim_{\lm\to0}c(\lm)=0$ such that for all $g\in \h^{S_{o,o'}}$, $z\in I+i[0,1]$,    $\lm\in[0,1]$, $w,w',\vartheta\in (-\lm, \lm)$
  $${\gm\ap{g_{o},\Gm_{o}(z,\Lp)}}
\leq(1+c(\lm))\sum_{x\in S_{o,o'}}p_{x}c_{x}(g)\gm_{x}(g)+c(\lm),$$
where $g_{o}=g_{o}(z,w,w',\vartheta,g)$ and $g_{o'}=g_{o'}(z,w',g)$.
\end{prop}
\begin{proof}
The statement follows directly from combining the statements for $o$ and $o'$ in the one step expansion estimate above.
\end{proof}


\section{The uniform contraction estimate}\label{s:UC}

In this section we define an averaged contraction coefficient and show that it is uniformly smaller than one. The definitions are inspired by the fact that the Green functions on two vertices are identically distributed if their labels coincide.


Let $\T$  again be a tree generated by a substitution matrix that satisfies $\mathrm{(M0)}$ and $\mathrm{(M1^*)}$. For a given vertex $o$ let $o'$ be a forward neighbor of $o$  chosen with respect to $\mathrm{(M1^*)}$.

We start with some simple, but important, facts that will later help to explain how we prove
the uniform contraction estimate.

\begin{lemma}\label{l:Z} Let $p\geq1$, $z\in\h$, $w\in\R$ and $g\in\h^{S_{o,o'}}$. Then
\begin{align*}
\Big({\sum_{x\in S_{o,o'}}p_{x}c_{x}(g)\gm_{x}(g)}\Big)^p
\leq \sum_{x\in S_{o,o'}}p_{x}\gm_{x}(g)^p
\leq \max_{x\in S_{o,o'}}\gm(g_{x},h_{x})^p, \end{align*}
Moreover, the following are equivalent
\begin{itemize}
\item [(i.)]   $\sum_{x\in S_{o,o'}}p_{x}c_{x}(g)\gm_{x}(g)=\sum_{x\in S_{o,o'}}p_{x}(g)\gm_{x}(g)$.
\item [(ii.)]  $c_{x}=1$ for all $x\in S_{o,o'}$.
\item [(iii.)]  $Q_{x,y}=1$ and $\cos\al_{x,y}=1$ for all $x,y\in S_{i}$ and $i\in\{o,o'\}$.
\end{itemize}
\begin{proof}
The statements follow from the Jensen inequality and the basic properties of the quantities such as that the $p_{x}$'s sum up to one and the $c_{y}$'s are bounded by one. \end{proof}\end{lemma}

Since the Jensen's inequality is generically strict, (i), (ii), (iii) are not equivalent to  equality in the first inequality of the lemma for $p>1$. This fact will be exploited in the proof of the uniform contraction estimate.


\subsection{Label invariant permutation and an averaged contraction coefficient}

Let $o\in V$. We introduce permutations of $S_{o,o'}$ that leave the label invariant.
\medskip

\begin{definition}\label{d:Pi}(Label invariant permutations) For $o\in V$ we define the set of \emph{label invariant permutations } $\Pi:=\Pi(S_{o,o'})$ of $S_{o,o'}$ by
\begin{align*}
\Pi:=\{\pi:S_{o,o'}\to S_{o,o'}\mid \mbox{bijective and $\pi(x)$ and $x$ carry the same label for all $x\in S_{o,o'}$ }\}.
\end{align*}
\end{definition}

For given $g\in\h^{S_{o,o'}}$, $\pi\in\Pi$ the composition of $g$ and $\pi$ is of course given by $g\circ\pi=(g_{\pi(x)})_{x\in S_{o,o'}}$. By the symmetry of the tree, we  get for unperturbed truncated Green function
$$\Gm_{S_{o,o'}}(z,\Lp)=\Gm_{S_{o,o'}}(z,\Lp)\circ\pi$$
for all $\pi\in\Pi$. As $\Gm_{x}(z,H^{\lm})$ and $\Gm_{y}(z,H^{\lm})$ are independent and identically distributed for $x,y\in S_{o,o'}$ that carry the same label, we see that
\begin{align*}
\EE(f(\Gm_{S_{o,o'}}(z,H^{\lm})))=    \EE(f(\Gm_{S_{o,o'}}(z,H^{\lm})\circ\pi))
\end{align*}
for any integrable function $f$ and all $\pi\in\Pi$.

For a function $f$ on $\h^{S_{o,o'}}$, we define $$f^{(\pi)}(g):=f(g\circ \pi).$$
This convention gives for example
\begin{align*}
g_{x}^{(\pi)}
=\left\{\begin{array}{ll}
g_{\pi(x)} &: x\in S_{o,o'}, \\
(g_{o'}\circ\pi)(z,w,g)=g_{o'}(z,w,g\circ\pi) &: x=o', \\
\end{array}
\right.
\end{align*}
and
\begin{align*}
\gm_{x}^{(\pi)}(g)
=\left\{\begin{array}{ll}
\gm(g_{\pi(x)},\Gm_{x}(z,\Lp)) &: x\in S_{o,o'}, \\
\gm(g_{o'}\circ\pi ,\Gm_{o'}(z,\Lp)) &: x=o'. \\
\end{array}
\right.
\end{align*}

\begin{definition}(Averaged contraction coefficient) Let $p\geq1$, $z\in\h$, $w\in\R$.
We define the function $\ka_o^{(p)}:=\ka_{o}^{(p)}(z,w,\cdot):\h^{S_{o,o'}}\to\R$, called the \emph{averaged contraction coefficient}, by
\begin{align*}
\ka_o^{(p)}(g)&:=\frac{\sum_{\pi\in\Pi}
\ab{\sum_{x\in S_{o,o'}}p_{x}c_{x}^{(\pi)}(g)\gm_{x}^{(\pi)}(g)}^p} {\sum_{\pi\in\Pi} \sum_{x\in S_{o,o'}}p_{x}c_{x}^{(\pi)}(g){\gm_{x}^{(\pi)}(g)}^{p}}.
\end{align*}
\end{definition}

The non-negativity of $\ka_{o}^{(p)}$ is an easy consequence of the two step expansion estimate. By Lemma~\ref{l:Z}, we also get that $\ka_o^{(p)}\leq 1$.
The aim  of this section is to prove
$\ka_o^{(p)}\leq 1-\de$ under suitable conditions as stated below.\medskip

\begin{prop}\label{p:ka}(Uniform contraction estimate)
Let $I\subset \Sigma$ be compact and $p>1$. There exist $\de_o=\de_o(I,p)>0$, $\lm_o=\lm_o(I)>0$ and $R_{o}:[0,\lm_o]\to[0,\infty)$ with $\lim_{\lm\to 0}R_{o}(\lm)=0$ such that for all $\lm\in[0,\lm_o]$
\begin{align*}
\sup_{z\in I+i[0,1]}\sup_{w\in[-\lm,\lm]}\sup_{g\in \h^{S_{o,o'}}\setminus B_{R_o(\lm)}}\ka_{o}^{(p)}(z,w,g)\leq 1-\de_o,
\end{align*}
where $B_r:=\{g\in\h^{S_{o,o'}}\mid \gm(g_{x},\Gm_{x}(z,\lm))\leq r,\mbox{ for all } x\in S_{o,o'}\}$ for $r\geq0$.
\end{prop}


\subsection{Strategy of the proof of uniform contraction}\label{ss:kaformula}
The proof of Proposition~\ref{p:ka} involves a thorough analysis of the quantities which enter the definition of $\ka_o^{(p)}$.

Before we discuss the structure of the proof let us fix some notation. For the remainder of the section we always assume that $I\subset\Sigma$ is compact and $z$ is an element of $ I+i[0,1]$. For short, we denote
\begin{align*}
    \Gm_{x}=\Gm_{x}(z,\Lp),\qquad x\in S_{o}\cup S_{o'}.
\end{align*}
Furthermore, we denote by $g$ always a vector in $\h^{S_{o,o'}}$ and if not stated otherwise $\pi$ denotes a label invariant permutation. Moreover,  $w$ always denotes an  real number.

Let us discuss the strategy of the proof. Our aim is to show that there is $\de>0$ such that $\ka_o^{(p)}\leq 1-\de$ outside of a compact set.
Lemma~\ref{l:Z} indicates two ways to prove uniform contraction.
The first one is to estimate
$$\Big({\sum_{x\in S_{o,o'}}p_{x}c_{x}^{(\pi)}(g)\gm_{x}^{(\pi)}(g)}\Big)^p\leq \Big({\sum_{x\in S_{o,o'}}p_{x}\gm_{x}^{(\pi)}(g)}\Big)^p\leq (1-\de)\sum_{x\in S_{o,o'}}p_{x}\gm_{x}^{(\pi)}(g)^p,$$
where the $(1-\de)$ is squeezed out of the error term of the Jensen inequality. Hence, $\de$ will be greater than zero whenever the $\gm_{x}^{(\pi)}(g)$ are of sufficiently different magnitude.
Secondly, one can try to find suitable $x$, $y$ and  $\pi$ such that the contraction quantities $Q_{x,y}^{(\pi)}(g)$ or $\cos\al_{x,y}^{(\pi)}(g)$ are less than one. Then, the corresponding $c_{x}^{(\pi)}(g)$ is smaller than one which implies $\ka_o^{(p)}(g)<1$.
However,  for uniform contraction we need more. Even if we manage to show uniform bounds for the contraction quantities $\cos\al_{x,y}^{(\pi)}(g)$ or $Q_{x,y}^{(\pi)}(g)$, (i.e., that one of them is less or equal to $c$ for some $c<1$), two problems can occur that make this contraction `invisible'. The first problem is that the weight $q_{y}^{(\pi)}(g)$ might  not be bounded from below and thus the contraction quantity $c_{x}^{(\pi)}(g)$ can become arbitrary close to one. The second problem is that even if one has $c_{x}^{(\pi)}\leq c$ for some $c<1$ the quantity $p_x{\gm_{x}^{(\pi)}(g)}^{p}/\sum_{\pi}\sum_{y} p_{y}{\gm_{y}^{(\pi)}(g)}^{p}$ might become so small such that $\ka_{o}^{(p)}(g)$ becomes arbitrary close to one. Our strategy to quantify these problems is to introduce the notion of visibility in Subsection~\ref{s:Vis}.

We prove Proposition~\ref{p:ka} by distinguishing three cases for $g\in\h^{S_{o,o'}}$.

In Case~1, Subsection~\ref{ss:Case1}, we look at the case where $\gm_{x}(g)/\max_{y\in S_{o,o'}}\gm_y(g)$ is very small for some $x\in S_{o,o'}$. This implies that $p_x\gm_{x}(g)/\sum_{\pi}\sum_{y}p_{y}\gm_{y}^{(\pi)}(g)^{p}$ is very small. In this case we get uniform contraction from the error term of Jensen's inequality.

In Case~2, Subsection~\ref{ss:Case2}, we assume that we are not in Case~1, but that there is $\pi\in\Pi$ and $x\in S_{o'}$ such that $\Im g_{x}^{(\pi)}/\max_{y\in S_{o'}}\Im g_{y}^{(\pi)}$ is very small. Hence, $q_{x}^{(\pi)}(g)$ is very small and we exploit that $Q_{x,y}(g)$ is the quotient of a geometric and an arithmetic mean.

For Case~3, Subsection~\ref{ss:Case3}, we prove that there are always $\pi$, $x$, $y$ such that $Q_{x,y}^{(\pi)}(g)\cos\al_{x,y}^{(\pi)}(g)$ is uniformly smaller than one. Assuming that we are not in the Case~1 or~2, none of the problems mentioned above occurs and we conclude  the uniform contraction estimate in this case.

Finally,  in Subsection~\ref{ss:kaproof}, we put the pieces together.


\subsection{Visibility}\label{s:Vis}

As discussed above, it makes only sense to look for those $\cos\al_{x,y}^{(\pi)}(g)$ and $Q_{x,y}^{(\pi)}(g)$, where the corresponding weights are uniformly bounded from below. Otherwise, the contraction might become `invisible'. This is quantified below by the sets of visible vertices.
\medskip

\begin{definition}\label{d:Vis}(Visibility)
For $g\in \h^{S_{o,o'}}$ and $\eps>0$, we define the set of vertices in $S_{o,o'}$ that are \emph{visible with respect to $\gm$ } by
\begin{align*}
\Vis_{\gm}(g,\eps)&:=\Big\{x\in S_{o,o'}\mid \min_{y\in S_{o,o'}}\frac{\gm_x(g)}{\gm_y(g)}>\eps\Big\}
\end{align*}
and the  set of vertices in $S_{i}$, $i\in\{o,o'\}$, that are \emph{visible with respect to the imaginary parts} by
\begin{align*}
\Vis_{\Im}^{i}(g,\eps)&:=\Big\{x\in S_{i}\mid \min_{y\in S_{i}}\frac{\Im g_x}{\Im g_y}>\eps\Big\},
\end{align*}
where $g_{o'}$ is defined in Section~\ref{ss:OSE}.
\end{definition}

Let us remark some simple facts. Firstly,
for $0<\eps'\leq\eps$ we have
$\Vis_{\gm}(g,\eps)\subseteq\Vis_{\gm}(g,\eps')$ and $\Vis_{\Im}^{i}(g,\eps)\subseteq\Vis_{\Im}^{i}(g,\eps')$, $i\in\{o,o'\}$.
Secondly, we have $\Vis_{\gm}(g,\eps)=\emptyset$ and/or $\Vis_{\Im}^{i}(g,\eps)=\emptyset$ for $i\in\{o,o'\}$ if and only if $\eps\geq1$. Finally, if $\Vis_{\gm}(g,\eps)=S_{o,o'}$ then $\Vis_{\gm}(g\circ\pi,\eps)=S_{o,o'}$ for all $\pi\in\Pi$.

We now prove two lemmas that demonstrate how we put the concept of visibility into action.
We start by visibility with respect to the imaginary parts.
\medskip

\begin{lemma}\label{l:using_alQ<1}
Let $\eps>0$, $c\in[0,1)$. There is $\de=\de(\eps,c)>0$ such that if $g\in\h^{S_{o,o'}}$ satisfies
\begin{align*}
Q_{\ow x,\ow y}^{(\pi)}(g)\cos\al_{\ow x,\ow y}^{(\pi)}(g)\leq c,
\end{align*}
for some $\pi\in\Pi$, $i\in\{o,o'\}$, $\ow x\in S_{i}$ and  $\ow y\in \Vis_{\Im}^{i}(g\circ\pi,\eps)$,
then it follows that
\begin{align*}
c_{\ow x}^{(\pi)}(g)\leq 1-\de.
\end{align*}
\begin{proof}
Note that, for all $\pi\in\Pi$, $i\in\{o,o'\}$ and $x\in\Vis_{\Im}^{i}(g\circ\pi,\eps)$, we have
\begin{align*}
q_{x}^{(\pi)}(g)
=\ab{\sum_{y\in S_{i}}\frac{\Im g_{y}^{(\pi)}}{\Im g_{x}^{(\pi)}}}^{-1}
\geq\frac{\eps}{|S_{i}|}.
\end{align*}
Let $\ow x,\ow y\in S_{i}$ be such that $Q_{\ow x,\ow y}^{(\pi)}(g)\cos\al_{\ow x,\ow y}^{(\pi)}(g)\leq c$.
We get by estimating $Q_{\ow x,y}^{(\pi)}(g)\cos\al_{\ow x,y}^{(\pi)}(g) \leq 1$, for all $y\neq \ow y$, that
\begin{align*}
c_{\ow x}^{(\pi)}(g)&\leq \sum_{y\in S_{i}\setminus\{\ow y\}}q_{y}^{(\pi)}(g)+q_{\ow y}^{(\pi)}(g)Q^{(\pi)}_{\ow x, \ow y}(g)\cos\al^{(\pi)}_{\ow x, \ow y}(g)\\
&=1-q_{\ow y}^{(\pi)}(g)\big(1-Q^{(\pi)}_{\ow x, \ow y}(g)\cos\al^{(\pi)}_{\ow x,\ow y}(g)\big)\leq1-\de,
\end{align*}
where we chose $\de=(1-c)\eps/\min_{i\in\{o,o'\}}|S_{i}|$.
\end{proof}
\end{lemma}

The next lemma shows how to use visibility with respect to $\gm$.\medskip

\begin{lemma}\label{l:using_c<1}
Let  $\eps>0$, $c\in[0,1)$. There is $\de=\de(\eps,c)>0$ such that if $g\in\h^{S_{o,o'}}$ satisfies
\begin{align*}
    \Vis_{\gm}(g,\eps)=S_{o,o'}
\qqand
c_{x}^{(\pi)}(g)\leq c,
\end{align*}
for some $\pi\in\Pi$ and  $x\in S_{o,o'}$, then it follows that
\begin{align*}
\ka_o^{(p)}(z,w,g)\leq 1-\de.
\end{align*}
\begin{proof}
As  $p_{x}$, $x\in S_{o,o'}$, are functions of the imaginary parts of the unperturbed Green functions, they are uniformly larger than zero for all $z$ in $I+i[0,1]$ (as $I$ is a compact subset of $\Sigma$). We conclude the existence of $\eps'>0$ such that for all $\pi\in\Pi$ and $x\in\Vis_{\gm}(g\circ\pi,\eps)$
\begin{align*}
\frac{p_{x}{\gm_{x}^{(\pi)}(g)}^p} {\sum_{\pi}\sum_{y}p_{y}{\gm^{(\pi)}_{y}(g)}^{p}}  =\ab{\sum_{\pi\in\Pi} \sum_{y\in S_{o,o'}}\frac{p_{y}}{p_{x}} \ab{\frac{\gm_{y}^{(\pi)}(g)}{\gm_{x}^{(\pi)}(g)}}^p}^{-1}
\geq\eps^p\ab{\sum_{\pi\in\Pi} \sum_{y\in S_{o,o'}}\frac{p_{y}}{p_{x}}}^{-1}\geq \eps'.
\end{align*}
Let $\pi\in\Pi$ and $x\in S_{o,o'}$  be chosen according to the assumption (which implies, in particular $x\in \Vis_{\gm}(g\circ\pi,\eps)$).
We get by estimating $c_{y}^{(\pi)}(g)\leq1$ for $y\neq  x$, $c_{x}^{(\pi)}(g)\leq c$, Jensen's inequality and the previous estimate
\begin{align*}
\Big(\sum_{x\in S_{o,o'}}p_{x}\gm^{(\pi)}_{x}(g)c_{x}^{(\pi)}(g)\Big)^p &\leq {\Big(\sum_{y\in S_{o,o'}\setminus\{{x}\}} p_{y}\gm_{y}^{(\pi)}(g) +p_{{x}}^{(\pi)}\gm_{{x}}(g)c\Big)}^{p}\\
&\leq \sum_{y\in S_{o,o'}\setminus\{{x}\}} p_{y}\gm_{y}^{(\pi)}(g)^p +p_{{x}}\gm_{{x}}^{(\pi)}(g)^{p}c^p\\
&= \sum_{y\in S_{o,o'}} p_{y}\gm_{y}^{(\pi)}(g)^{p} -p_{{x}}\gm_{{x}}^{(\pi)}(g)^{p}(1-c^p)\\
&\leq\sum_{y\in S_{o,o'}} p_{y}\gm_{y}^{(\pi)}(g)^{p}
-\eps'(1-c^p)\sum_{\ow\pi\in\Pi}\sum_{y\in S_{o,o'}} p_{y}\gm_{y}^{(\ow\pi)}(g)^{p}.
\end{align*}
We set $\de=\eps'(1-c^p)$. Applying the basic  estimates of Lemma~\ref{l:Z}, for all $\ow \pi\in\Pi$, $\ow \pi\neq\pi$ and using the inequality above, we get
\begin{align*}
\ka_{o}^{(p)}(g) \leq\frac{{ \sum_{\ow\pi\neq\pi}\sum_{y\in S_{o,o'}} p_{y}\gm_{y}^{(\ow\pi)}(g)^{p} +\Big(\sum_{x\in S_{o,o'}} p_{x}\gm^{(\pi)}_{x}(g)c_{x}^{(\pi)}(g)\Big)^p }}{\sum_{\ow\pi\in\Pi}\sum_{y\in S_{o,o'}}  p_{y}\gm_{y}^{(\ow\pi)}(g)^{p} } \leq 1-\de.
\end{align*}
\end{proof}
\end{lemma}


\subsection{Case~1: A Jensen type inequality}\label{ss:Case1}
In this section we will prove the statement that was discussed as Case~1 in Subsection~\ref{ss:kaformula}.
We will need the following auxiliary constant
\begin{equation*}\label{e:c1}
c_1:=c_1(I):=\Big({\max_{z\in I+i[0,1]}\max_{x,y\in S_{o,o'}}\frac{(1-p_{x})}{p_{y}}}\Big)^{-1}. \end{equation*}
Since $p_{x}\in(0,1)$ and $p_{x}$ depend continuously on $z$ and $I\subseteq\Sigma$ is compact, $c_1$ is finite and positive.

\medskip

\begin{prop}\label{p:invisgm}
Let $p>1$, $\eps\in(0,c_1)$.  There is $\de=\de(\eps)>0$  such that if  $g\in \h^{S_{o,o'}}$ satisfies
\begin{align*}
\Vis_{\gm}(g,\eps)\neq S_{o,o'},
\end{align*}
then it follows that
\begin{align*}
\ka_o^{(p)}(z,w,g)\leq 1-\de.
\end{align*}
\end{prop}

For the proof we employ a refinement of Jensen's inequality for monomials. To this end, we take a closer look at the error term in Jensen's inequality.
\medskip

\begin{lemma}Let $f:\R\to\R$ be twice continuously differentiable, $\lm\in[0,1]$ and $r,s\in\R$. Then
\begin{align*}
f(\lm r+(1-\lm)s)=\lm f(r)+(1-\lm)f(s)-E_f(\lm,r,s),
\end{align*}
where
\begin{align*}
E_f(\lm,r,s)=(r-s)^2\!\int_0^1\!\! \ap{\lm t \mathds{1}_{[0,1-\lm]}(t)+(1-\lm)(1-t)\mathds{1}_{[1-\lm,1]}(t)}f''((1-t)r+ts)dt    \end{align*}
and $\mathds{1}_A$ is the characteristic function of  a set $A$.
\begin{proof} We take the Taylor expansion of $f$ in $r_0=\lm r+(1-\lm)s$  at the points $r$ and $s$ with integral error term. Inserting this into $\lm f(r)+(1-\lm)f(s)$ yields the statement.
\end{proof}
\end{lemma}

Jensen's inequality for twice continuously differentiable functions is a direct corollary. However,  if we know more about the function $f$ we can use the error term to get finer estimates. In our case $f$ is the function $r\mapsto r^{p}$, $p> 1$.\medskip

\begin{lemma}\label{l:Jensen2}Let $p\geq1$, $\lm\in[0,1]$ and $r,s\in[0,\infty)$, $r> s$. Then
\begin{align*}
(\lm r+(1-\lm)s)^p\leq(1-\de_p)\ap{\lm r^p+(1-\lm)s^p},
\end{align*}
where
\begin{align*}
\de_p:=\de_p(\lm,s/r):=\ap{1-{s}/{r}}^2\left\{\begin{array}{ll}
p(p-1)\lm(1-\lm)/2&: p\in[1,2), \\
\lm(1-\lm^{p-1})&: p\geq2.
\end{array}\right.
\end{align*}
\begin{proof}
The statement is trivial for $p=1$ and $\lm\in\{0,1\}$. Therefore,  we assume $p>1$ and $\lm\not\in\{0,1\}$.

Assume first that $r=1$ and $s\in[0,1)$. Then, the error term $E_p:=E_{(\cdot)^p}$ from the previous lemma can be estimated for $p\geq2$, using $s\geq0$, as follows
\begin{align*}
\frac{E_p(\lm,1,s)}{(1-s)^2}&=p(p-1)\int_0^1 \ap{\lm t \mathds{1}_{[0,1-\lm]} +(1-\lm)(1-t)\mathds{1}_{[1-\lm,1]}}((1-t)+ts)^{p-2}dt\\
&\geq p(p-1)\ap{\lm\int_{0}^{1-\lm}t(1-t)^{p-2}dt+ (1-\lm) \int_{1-\lm}^{1}(1-t)^{p-1}dt}\\
&=\lm\ap{1-\lm^{p-1}}.
\end{align*}
On the other hand, for $p\in(1,2)$ we get by estimating $((1-t)+ts)^{(p-2)}\geq1$ that
\begin{align*}
\frac{E_p(\lm,1,s)}{(1-s)^2}&\geq p(p-1)\int_0^1 \ap{\lm t \mathds{1}_{[0,1-\lm]}+(1-\lm)(1-t)\mathds{1}_{[1-\lm,1]}}dt =p(p-1)\frac{\lm(1-\lm)}{2}.
\end{align*}
Employing the previous lemma  for $r=1$, $s\in[0,1]$, we get since $\lm +(1-\lm) s^p\leq 1$,
\begin{align*}
(\lm +(1-\lm)s)^p&\leq{\lm +(1-\lm)s^p}-\de_p\leq(1-\de_p)\ap{\lm +(1-\lm)s^p}.
\end{align*}
Now, let $r> s\geq0$ be arbitrary. By the previous inequality, we obtain
\begin{align*}
(\lm r+(1-\lm)s)^p&=r^p\ap{\lm +(1-\lm)\frac{s}{r}}^p\leq (1-\de_p)\ap{\lm r^p+(1-\lm)s^p}.
\end{align*}
\end{proof}
\end{lemma}

With these preparations we are ready to prove Proposition~\ref{p:invisgm}.

\begin{proof}[Proof of Proposition~\ref{p:invisgm}]
Let  first $x_0\in S_{o,o'}$, $\pi\in \Pi$ be arbitrary. We apply the previous lemma with $\lm=(1-p_{x_0})$, $r=\sum_{x\neq x_0}\tfrac{p_{x}}{(1-p_{x_0})}\gm^{(\pi)}_{x}(g) $ and $s={\gm_{x_0}^{(\pi)}(g)}$ to obtain
\begin{align*}
\Big({\sum_{x\in S_{o,o'}}p_{x}\gm_{x}^{(\pi)}(g)}\Big)^p
\leq(1-\de_p)\sum_{x\in S_{o,o'}}p_{x}{\gm_{x}^{(\pi)}(g)}^p,
\end{align*}
where we also applied the Jensen inequality to estimate $\lm r^{p}\leq\sum_{x\neq x_{0}}p_{x}\gm_{x}^{(\pi)}(g)^{p}$. By estimating $c_{x}^{(\pi)}(g)\leq1$ for all $x\in S_{o,o'}$, we get
\begin{align*}
\ka_o^{(p)}(g) \leq\frac{\sum_{\pi\in\Pi}\ap{\sum_{x\in S_{o,o'}} p_{x}\gm_{x}^{(\pi)}(g)}^p}{\sum_{\pi\in\Pi} \sum_{x\in S_{o,o'}}p_{x}{\gm_{x}^{(\pi)}(g)}^p}\leq 1-\de_p.
\end{align*}
So, far we have not said anything about the value of $\de_{p}$. This will be done to finish the proof. The constant $\de_{p}$ is a product of a term involving $s/r$ and one involving $p$ and $\lm$, see Lemma~\ref{l:Jensen2}. Recalling that we chose $\lm=1-p_{x_0}$ above, we estimate the part involving $\lm$ and $p$ by the quantity $c_2:=c_2(I,p)$ defined as
\begin{align*}
c_{2}&:=\min_{z\in I+i[0,1]}\min_{x\in S_{o,o'}}(1-p_{x}) \min\set{\frac{p(p-1)}{2}p_{x},\ap{1-(1-p_{x})^{p-1}}}.
\end{align*}
As the $p_x$'s are uniformly larger than zero and $p>1$, we have $c_2>0$. For given $\pi\in\Pi$ let now $x_0,\oh x\in S_{o,o'}$ be chosen such that $\gm_{x_0}^{(\pi)}(g)=\min_{x\in S_{o,o'}}\gm_{x}^{(\pi)}(g)$, $\gm_{\oh x}^{(\pi)}(g)=\max_{x\in S_{o,o'}}\gm_{x}^{(\pi)}(g)$. This is the part where the assumption $ \Vis_{\gm}(g,\eps)\neq S_{o,o'}$ comes into play. By assumption  $x_0$ must be in $S_{o,o'}\setminus \Vis_{\gm}(g\circ\pi,\eps)$ and, hence, $\gm_{x_0}^{(\pi)}(g)/\gm_{\oh x}^{(\pi)}(g)\leq \eps$.
Now, with $r=\sum_{x\neq x_{0}} \frac{p_{x}}{(1-p_{x_0})}\gm_{x}^{(\pi)}(g)$ and $s=\gm_{x_0}^{(\pi)}(g)$ as above,
we obtain by the definition of $c_1$
\begin{align*}
\frac{r}{s}=\frac{(1- p_{x_0})\gm_{x_0}^{(\pi)}(g)} {\sum_{x} p_{x}\gm_{x}^{(\pi)}(g)}
= \frac{\gm_{x_0}^{(\pi)}(g)}{\gm_{\oh x}^{(\pi)}(g)}\frac{(1- p_{x_0})} {\sum_{x} p_{x}\gm_{x}^{(\pi)}(g)/\gm_{\oh x}^{(\pi)}(g) }
\leq \eps\frac{(1-p_{x_{0}})}{p_{\oh x}}\leq \frac{\eps}{c_1}.
\end{align*}
Hence, $
\de_{p}\geq c_2\ap{1-{\eps}/c_1}^2>0$ by Lemma~\ref{l:Jensen2} and we finished the proof.
\end{proof}

\subsection{Case 2: Geometric and arithmetic means}\label{ss:Case2}

This subsection deals with Case~2 as discussed in Subsection~\ref{ss:kaformula}.  Define
\begin{equation*}\label{e:eps0}
\eps_0:=\min\set{\frac{\Im \Gm_x}{\Im \Gm_{y}}\mid {z\in I+i[0,1]}, {x,y\in S_{o,o'}}}.
\end{equation*}
As $I\subset\Sigma$ is  compact, we have $\eps_0>0$.
\medskip

\begin{prop}\label{p:invisIm}
Let $\eps>0$, $\eps'\in(0,\eps \eps_0)$. There is $\de=\de(\eps,\eps')>0$ such if $\pi\in\Pi$ and $g\in \h^{S_{o,o'}}$  satisfy
\begin{align*}
    \Vis_{\gm}(g,\eps)=S_{o,o'}
\qqand
\Vis_{\Im}^{o'}(g\circ\pi,\eps')\neq S_{o'},
\end{align*}
then it follows that there exists $x\in S_{o'}$ such that
\begin{align*}
c_{x}^{(\pi)}(g)\leq (1-\de).
\end{align*}
\begin{proof}
Since $Q_{x,y}(g)$ is the ratio of a geometric and an arithmetic mean it is equal to one if and only if the quantities that enter are equal. For a finer analysis we introduce, for  $x,y\in S_{o'}$,
\begin{align*}
\varrho_{x,y}(g) :=\frac{\Im g_x\Im \Gm_{y}\gm_y(g)}{\Im g_y\Im \Gm_{x}\gm_x(g)}.
\end{align*}
Observe that $\varrho_{x,y}=1/\varrho_{y,x}$. Then
\begin{align*}
Q_{x,y}(g)= \frac{2}{\ap{\sqrt{\varrho_{x,y}(g)} +\sqrt{\varrho_{y,x}(g)}}}.
\end{align*}
We have that $Q_{x,y}(g)<1$ if and only if $\varrho_{x,y}(g)\neq 1$.\\
Let $\pi\in\Pi$ be such that $\Vis_{\Im}^{o'}(g\circ\pi,\eps')\neq S_{o'}$. This implies the existence of $x,y\in S_{o'}$ with ${\Im g_{y}^{(\pi)}}/{\Im g_{x}^{(\pi)}}\leq\eps'$.
Since by assumption $x, y\in \Vis_{\gm}(g\circ\pi,\eps)$, we have $\gm_{y}^{(\pi)}(g)/\gm_{x}^{(\pi)}(g)>\eps$.
Hence,
\begin{align*}
\varrho_{x,y}(g\circ\pi)>\frac{\eps_{0}\eps}{\eps'}.
\end{align*}
Employing the inequality $r/(1+r^2)\leq s/(1+s^2)$ for $r\geq s\geq1$, we get
\begin{align*}
Q_{x,y}(g\circ\pi) =\frac{2\sqrt{\varrho_{x,y}(g\circ\pi)}}{1+\varrho_{x,y}(g\circ\pi)} \leq \frac{2\sqrt{\eps_0\eps \eps'}}{\eps_0 \eps+\eps'}<1.
\end{align*}
By the choice of $\eps'\in(0,\eps_0\eps)$, we get  $Q_{x,y}(g\circ\pi)<c$ with $c=2\sqrt{\eps_{0}\eps\eps'}/(\eps_{0}\eps+\eps')<1$.
The statement now follows from Lemma~\ref{l:using_alQ<1}.
\end{proof} \end{prop}

\subsection{Case 3: A general bound on the relative arguments}\label{ss:Case3}

Until this point we used the definition of the argument of a complex number only in combination with the cosine. Now, we want to consider it as a function itself, so a little more care has to be taken in its particular definition.

We consider the argument of a non zero complex number as the continuous group homomorphism $\arg:\C\setminus\{0\}\to\Sp^{1}\cong\R/2\pi\Z$,  where in this context $\pi$ denotes of course the number $\pi$. So, whenever we speak in the following about permutations we indicate them by saying explicitly that they are elements of $\Pi$.
Moreover, we denote by $d_{\Sp^{1}}(\cdot,\cdot)$ the translation invariant metric in $\Sp^{1}$ which is normalized by $d_{\Sp^{1}}(0,\pi)=\pi$.

We define a quantity, related to the `minimal angle' of the unperturbed Green function with the real axis, by
\begin{align*} \de_0:=\frac{1}{4}\min\set{d_{\Sp^{1}}(\arg \Gm_x,\be)\mid \be\in\{0,\pi\}, z\in I+i[0,1], x\in S_{o,o'}}.
\end{align*}
Since $I\subset\Sigma$ is chosen compact, the minimum exists and we have $\de_0>0$. Recall also the definition $B_r=\{g\in\h^{S_{o,o'}}\mid \gm(g_{x},\Gm_{x})\leq r, \mbox{ for all } x\in S_{o,o'}\}$ for $r\geq0$ and also the definition of $\eps_{0}$ from the previous section.
\medskip

\begin{prop}\label{p:al>0}
There is $c=c(\de_0)<1$, $\lm_0=\lm_0(\de_0,\eps_0)>0$ and $R:[0,\lm_0)\to[0,\infty)$ with $\lim_{\lm\to0}R(\lm)=0$ such that for all $\lm\in[0,\lm_0)$, $w\in[-\lm,\lm]$, $g\in \h^{S_{o,o'}}\setminus B_{R(\lm)}$ there is $\pi\in \Pi$ with
\begin{align*}
Q_{x,y}^{(\pi)}\leq c\qquad\mbox{or}\qquad
\cos    \al_{x,y}^{(\pi)}\leq c,
\end{align*}
either for  some  $x,y\in S_{o'}$ or for $x=o'$ and all $y\in S_{o}$.
\end{prop}

Note that the $Q_{x,y}\leq c$ is in the statement of the proposition only to deal with the trivial cases $g_{x}=\Gm_{x}$ or $g_{y}=\Gm_{y}$. Before we come to the proof we need some basic geometric observations. The first one is about perturbations of arguments. \medskip

\begin{lemma}\label{l:arg} Let $\xi,\zeta\in\C$, such that $d_{\Sp^{1}}(\arg\xi,0)\leq \pi/2$ and $\mo{\zeta}<1$.
Then
\begin{align*}
d_{\Sp^{1}}(\arg\ap{1+ \xi+\zeta},0)\leq
\left\{
  \begin{array}{ll}
d_{\Sp^{1}}(\arg \xi,0) +\frac{|\zeta|}{1-|\zeta|}    &:\xi\neq0, \\
\frac{|\zeta|}{1-|\zeta|}&: \xi=0 .\\
  \end{array}
\right.
\end{align*}
\begin{proof}
We write $\xi$, $\zeta$ in polar coordinates $\xi=r e^{i\de}$ and $\zeta=\eps e^{i\te}$. We denote the left hand side of the inequality by $\be$, i.e., $\be=d_{\Sp^{1}}(\arg\ap{1+ r e^{i\de}+\eps e^{i\te}},0)$. Assume without loss of generality that $\de\ge0$. Since we have for $z\in \C$ with $\Re z,\Im z\geq0$ and $u\geq0$ that $d_{\Sp^{1}}({\arg{(z+iu)},0})\geq d_{\Sp^{1}}({\arg(z-iu ),0})$ we conclude $\be\leq d_{\Sp^{1}}({\arg{(1+ r e^{i\de}+\eps e^{i|\te|})},0})$. Hence, we may assume without loss of generality that $\te\ge0$. We  use $\arg z=\arctan (\Im z/\Re z)$, subadditivity and monotonicity of $\arctan$ on $[0,\infty)$  and  $0\leq\arctan'\leq1$ to calculate
\begin{align*}
\be&=\arctan\ap{\frac{r\sin\de+\eps\sin\te}{1+r\cos\de+\eps\cos\te}}\\
&\leq\arctan\ap{\frac{r\sin\de}{1+r\cos\de+\eps\cos\te}}+
\arctan\ap{\frac{\eps\sin\te}{1+r\cos\de+\eps\cos\te}}\\
&\leq \arctan\ap{\frac{\sin\de}{\cos\de}}+ \arctan\ap{\frac{\eps}{1-\eps}} \leq\de+ \frac{\eps}{1-\eps}.
\end{align*}
The statement about $\xi=0$ is a direct consequence from the first statement.
\end{proof}
\end{lemma}

The second auxiliary lemma deals with sums of complex numbers.\medskip

\begin{lemma}\label{l:sum} Let $\de\in[0,\pi/2]$, $\xi\in\C^{S_{o,o'}}$ with $\xi_x\neq0$, $x\in S_{o,o'}$,  and $d_{\Sp^{1}} \big(\arg(\xi_x^{(\pi)}),\arg(\xi_y^{(\pi)})\big)\leq\de$ for all $\pi\in\Pi$, $x,y\in S_{o'}$. Then,
\begin{itemize}
\item [(1.)] $d_{\Sp^{1}}(\arg(\xi_x),\arg(\xi_y))\leq2\de$ for all $x,y\in S_{o,o'}$,
\item [(2.)]  $\big|{\sum_{y\in S_{o'}}\xi_{y}^{(\pi)}}\big|\geq\big|{\xi_{x}^{(\pi)}}\big|$ for all $\pi\in\Pi$ and $x\in S_{o'}$,
\item [(3.)] $d_{\Sp^{1}}\big({\arg\big({\sum_{y\in S_{o'}}\xi_{y}^{(\pi)}}\big), \arg(\xi_x^{(\pi)})}\big)\leq2\de$ for all $\pi\in\Pi$ and $x\in S_{o,o'}$.
\end{itemize}
\begin{proof}
The numbers $\xi_{x}$, $x\in S_{o,o'}$ can be thought as non zero vectors in the complex plane. Then, the assumption about the arguments means that they  point approximately in the same direction. The first statement follows as we can compare two elements $\xi_{x}$, $\xi_{y}$ always over a third one. The existence of such an element is guaranteed by $\mathrm{(M1^*)}$. The second and the third statement can easily be  seen by a direct calculation (using Lemma~\ref{l:arg} for (3.)) or  simply by drawing a picture.
\end{proof}
\end{lemma}

We define the function $\eta_{1}:[0,\infty)\to[0,\infty)$ that measures the Euclidean distance of hyperbolic balls $B_{r}=\{g\in\h^{\A}\mid \gm(g_{x},\Gm_{x})\leq r,x\in S_{o,o'}\}$ to the real axis, by
\begin{align*}
\eta_{1}(r):=\inf\{|g_{x}-\Gm_{x}|\mid z\in I+i[0,1], {g\in \h^{S_{o,o'}}\setminus B_r}, x\in S_{o,o'}\}.
\end{align*}
The formula in the next lemma below describes the inverse function of $\eta_{1}$. Moreover we define $\eps_{1}$ by
\begin{equation*}
\eps_1:=\min\set{{\Im \Gm_x}\mid {z\in I+i[0,1]}, {x\in S_{o,o'}}}.
\end{equation*}

\begin{lemma}\label{l:eps1} The function $\eta_1$ takes values in $[0,\eps_1)$ and its inverse function $\eta_{1}^{-1}:[0,\eps_{1})\to[0,\infty)$ is given by
\begin{align*}
    \eta_{1}^{-1}(s)=\frac{s^{2}}{(\eps_{1}-s)\eps_{1}}
\end{align*}
\begin{proof}The proof is rather direct but involves some background on the geometry of hyperbolic balls in the upper half  plane. For all $x\in S_{o,o'}$, $z\in I+i[0,1]$ and $r\geq0$ there is a unique  $\xi\in\h$ with
\begin{align*}
|\xi-\Gm_{x}| =\min_{\zeta\in\h,\gm(\zeta,\Gm_x)=r}|\zeta-\Gm_{x}|,\qquad \Im \xi=\min_{\zeta\in\h,\gm(\zeta,\Gm_x)=r}\Im \zeta
\end{align*}
and
\begin{align*}
    \Re\xi=\Re \Gm_{x}\qand \Im \xi=\Im \Gm_{x}-|\xi-\Gm_{x}|.
\end{align*}
(For details on the proof of these simple facts see for example Section~2.3.4 in \cite{Kel}.)
Let $\eta_{2}:[0,\infty)\to(0,\infty)$ be given by $$\eta_2(r)=\min\{\Im g_{x}\mid z\in I+i[0,1], {g\in B_r}, x\in S_{o,o'}\}.$$
By the considerations above and compactness of $I+i[0,1]$, there are $x_0$, $z$ and $g$ for each $r\geq0$ such that  $\eta_1(r)=|g_{x_0}-\Gm_{x_0}|$ and $\eta_{2}(r)=\Im g_{x_0}$.
We conclude
$$\eta_{1}(r)+\eta_{2}(r)=\eps_{1}$$
which shows that $\eta_1$ takes values in $[0,\eps_1)$ since $\eta_{2}>0$.
Resolving the equation $\gm(g_{x_0},\Gm_{x_{0}})=r$ yields, using the definition of $\gm$,
\begin{align*}
r=\gm(g_{x_{0}},\Gm_{x_{0}})=\frac{|g_{x_{0}}-\Gm_{x_{0}}|^{2}}{\Im g_{x_{0}}\Im \Gm_{x_{0}}}=\frac{\eta_1(r)^{2}}{\eps_1\eta_{2}(r)} =\frac{\eta_1(r)^{2}}{(\eps_{1}-\eta_{1}(r))\eps_1}.
\end{align*}
\end{proof}
\end{lemma}

Before we  prove Proposition~\ref{p:al>0}, let us remark some simple facts about calculating with arguments. We have the following relation for $\xi,\zeta\in\C\setminus\{0\}$
\begin{align*}
    \arg(\xi\zeta)=\arg(\xi)+\arg(\zeta)=\arg(\xi)-\arg(\ov\zeta),
\end{align*}
where $\ov\zeta$ denotes the complex conjugate and the sum is of course considered in $\R\setminus 2\pi\Z$.
Moreover, for $\al,\be,\gm,\de\in \Sp^{1}$
\begin{align*}
    d_{\Sp^{1}}(\al+\be,\gm+\de)\leq d_{\Sp^{1}}(\al,\gm)+d_{\Sp^{1}}(\be,\de).
\end{align*}
Finally, let us mention that since $\al_{x,y}(g)=\arg((g_{x}-\Gm_{x})\ov{(g_{y}-\Gm_{y})})$ we have
\begin{align*}
d_{\Sp^{1}}(\al_{x,y}(g),0) =d_{\Sp^{1}}(\arg(g_{x}-\Gm_{x}),\arg(g_{y}-\Gm_{y})).
\end{align*}

\begin{proof}[Proof of Proposition~\ref{p:al>0}]
Let $R:[0,\lm_0)\to[0,\infty)$ with $\lm_{0}:={\eps_1\de_0}/\ap{1+\de_0}$ be defined as
\begin{align*}
R(\lm)=\eta_{1}^{-1}(\de(\lm)),\quad\mbox{with } \de(\lm)=\frac{(1+\de_0)}{\de_0}\lm.
\end{align*}

Take $g\in \h^{S_{o,o'}}\setminus B_{R(\lm)}$ and $\lm\in[0,\lm_{0})$. If there is $x\in S_{o,o'}$ such that $g_{x}=\Gm_{x}$, then $Q_{x,y}=0$ by definition for all $y\in S_{o,o'}$. In this case we are done. Therefore,  assume that $g_{x}\neq \Gm_{x}$ for all $x\in S_{o,o'}$.
Moreover, assume $d_{\Sp^{1}}({\al_{x,y}^{(\pi)},0})\leq\de_0$ for all $\pi\in\Pi$ and $x,y\in S_{o'}$ since otherwise we are also done. So, our aim is to show $d_{\Sp^{1}}({\al_{x,y}^{(\pi)},0})>\de_0$ for some $\pi\in\Pi$ and all $y\in S_{o}$.

We start with two claims. Set $\tau_{o'}^{(\pi)}:=\sum_{y\in S_{o'}}({g_{y}^{(\pi)}-\Gm_{y}})$.

\emph{Claim 1: There is  $\pi\in\Pi$ such that
$|\tau_{o'}^{(\pi)}| \geq \frac{(1+\de_0)}{\de_0}\lm.$}\\
Proof of Claim~1.
We assumed that $g\not\in B_{R(\lm)}$, so there is $\oh x\in S_{o,o'}$ such that $\gm_{\oh x}(g)\geq R(\lm)$. By the choice of $o'$ with respect to $\mathrm{(M1^*)}$ there is $\pi$ such that $\pi(\oh x)\in S_{o'}$.
By Lemma~\ref{l:sum}~(2.), the definitions of  $\eta_{1}$ and $R$ above and Lemma~\ref{l:eps1}, we obtain
\begin{align*}
|\tau_{o'}^{(\pi)}| =\Big|{\sum_{y\in S_{o'}}({g_{y}^{(\pi)}-\Gm_{y}}})\Big| \geq\mo{g_{\oh x}-\Gm_{\oh x}}\geq \eta_{1}(R(\lm))=\de(\lm)=\frac{(1+\de_0)}{\de_0}\lm.
\end{align*}

\emph{Claim 2:  There is  $\pi\in\Pi$ such that for all $x\in S_{o,o'}$}
\begin{align*}
d_{\Sp^{1}}\big({\arg(\tau_{o'}^{(\pi)}-w) ,\arg{(g_{x}^{(\pi)}-\Gm_{x})}}\big) \leq 3\de_0.
\end{align*}
Proof of Claim~2.  Note that $|w/\tau_{o'}^{(\pi)}|\leq \de_0/(1+\de_0)$  by the assumption $w\in[-\lm,\lm]$ and Claim~1.
We put $\xi_{x}^{(\pi)}={g_{x}^{(\pi)}-\Gm_{x}}$ and get
\begin{align*}
d_{\Sp^{1}}\big({\arg({\tau_{o'}^{(\pi)}-w}), \arg(\xi_{x})}\big) &\leq d_{\Sp^{1}}\big({\arg({1-w/\tau_{o'}^{(\pi)}}),0} \big) +d_{\Sp^{1}}\big({\arg{(\tau_{o'}^{(\pi)})}, \arg(\xi_{x})}\big)
\leq3\de_0,
\end{align*}
where we used the formulas $\arg(\xi+\zeta)=\arg(\xi)+\arg(1+\zeta/\xi)$ and $d_{\Sp^{1}}(\al+\be,\gm)\leq d_{\Sp^{1}}(\al,0)+d_{\Sp^{1}}(\be,\gm)$  in the first step and  Lemma~\ref{l:arg} (applied with $\xi=0$ and $\zeta=w/\tau_{o'}^{(\pi)}$)  and Lemma~\ref{l:sum}~(3.)  in the second step.

By the definitions $g_{o'}^{(\pi)}$ and \eqref{e:rec}, we get  for $\pi\in \Pi$
\begin{align*}
\arg\big({g_{o'}^{(\pi)}-\Gm_{o'}} \big) &=\arg\ab{\frac{-1}{z-v^{\per}(o')-w+\sum_{x\in S_{o'}} g_{x}^{(\pi)}}-\frac{-1}{z-v^{\per}(o')+\sum_{x\in S_{o'}} \Gm_{x}}}\\
&=\arg\big({g_{o'}^{(\pi)}\Gm_{o'}
({\tau_{o'}^{(\pi)} -w})}\big).
\end{align*}
By definition of $\de_0$ we have that
$d_{\Sp^{1}}(\arg \Gm_{o'},\be)\geq4\de_0$ for $\be\in\{0,\pi\}$.
Moreover, since $g_{o'}^{(\pi)}\in\h$, we also have $d_{\Sp^{1}}(\arg g_{o'}^{(\pi)},\be)>0$ for $\be\in\{0,\pi\}$. Hence, $d_{\Sp^{1}}({\arg({g_{o'}^{(\pi)}\Gm_{o'}})},0)>4\de_0$.
Combining this with  the equation above and Claim~2,  we get for the permutation $\pi\in\Pi$ taken from Claim~2 and all $y\in S_{o}$
\begin{align*}
d_{\Sp^{1}}\big({{\al_{o',y}^{(\pi)}}(g),0}\big)&= d_{\Sp^{1}}\ab{\arg({g_{o'}^{(\pi)}\Gm_{o'}}(\tau_{o'}^{(\pi)} -w)),\arg(g_{y}^{(\pi)}-\Gm_{y})}\\
&\geq d_{\Sp^{1}}\ab{\arg({g_{o'}^{(\pi)}\Gm_{o'}}),0}
-d_{\Sp^{1}}\ab{\arg({\tau_{o'}^{(\pi)}-w}), \arg{(g_{y}-\Gm_{y})}}>\de_0,
\end{align*}
where we used $d_{\Sp}(\al+\be,\gm)\geq d_{\Sp}(\al,0)-d_{\Sp}(\be,\gm)$.
The assertion follows by letting $c:=\cos\de_0$.
\end{proof}

\subsection{Proof of the uniform contraction estimate}\label{ss:kaproof}
In this  subsection we finally put the pieces together in order to prove Proposition~\ref{p:ka}.

\begin{proof}[Proof of Proposition~\ref{p:ka}]
Let $I\subset\Sigma$ be compact. Let $\eps_0=\eps_0(I)$, $\de_0=\de_0(I)$ be as defined in Subsection~\ref{ss:Case2} and~\ref{ss:Case3}. Moreover, let $R:[0,\lm_0)\to[0,\infty)$ and $\lm_0>0$ be given by Proposition~\ref{p:al>0}. Choose $\de_1\in(0,c_1)$, where $c_1$ is defined in Subsection~\ref{ss:Case1} and pick $\de_2\in(0,\eps_0 \de_1)$.

Let $\lm\in[0,\lm_0)$, $g\in \h^{S_{o,o'}}\setminus B_{R(\lm)}$, $w\in[-\lm,\lm]$  and $z\in I+i[0,1]$. We now consider the three cases which we already distinguished above:

\emph{Case~1: $\Vis_{\gm}(g,\de_1)\neq S_{o,o'}$.}
The statement  follows directly from Proposition~\ref{p:invisgm}.

\emph{Case~2: $\Vis_{\gm}(g,\de_1)=S_{o,o'}$ but $\Vis_{\Im}^{o'}(g\circ\pi,\de_2)\neq S_{o'}$ for some $\pi\in\Pi$.}
The statement  follows by  combining of Proposition~\ref{p:invisIm} and Lemma~\ref{l:using_c<1}.

\emph{Case~3: $\Vis_{\gm}(g,\de_1)=S_{o,o'}$ and $\Vis_{\Im}^{o'}(g\circ\pi,\de_2)= S_{o'}$ for all $\pi\in\Pi$.} We find  $\pi$, $i$, $\ow x$, $\ow y$ which satisfy the assumptions of Lemma~\ref{l:using_alQ<1} as follows: By Proposition~\ref{p:al>0}  there is $\pi\in\Pi$ and $c<1$ such that  $Q_{x,y}^{(\pi)}\cos\al_{x,y}^{(\pi)}\leq c$ either for some $x,y\in S_{o'}$ or for some $x=o'$ and all $y\in S_{o}$.\\
In the first case we have $y\in\Vis_{\Im}^{o'}(g\circ\pi,\de_2)$ by assumption. We choose consequently $i=o'$, $\ow x=x$ and $\ow y=y$. \\
For the second case let $u\in \Vis_{\Im}^{o}(g\circ\pi,\de_2)$. We set $i=o$, $\ow y=u$ and  pick $\ow x\in S_{o}\setminus\{u\}$ arbitrary.\\
With these choices we see by Lemma~\ref{l:using_alQ<1} that the assumptions of Lemma~\ref{l:using_c<1} are satisfied and the statement follows.
\end{proof}


\section{A vector inequality}\label{s:VI}

In this section we prove the crucial inequality from which we deduce Theorem~\ref{main3}.

Let $\A$ be a finite set and $M:\A\times\A\to\N_0$ satisfying $\mathrm{(M0)}$ and $\mathrm{(M1^*)}$. We denote by $\T_j$ the tree generated by $M$ whose root $o(j)$ carries label $j\in\A$ and by $V_j$ the vertex set of $\T_{j}$.

Let the stochastic matrix $P=P(z):{\A\times\A}\to[0,\infty)$ for $z\in \h\cup\Sigma$ be given by
\begin{align*}
P_{j,k}:=\sum_{{x\in S_{o(j),o(j)'},\,}\atop{x\mbox{\scriptsize{ carries label }}k}}p_{x},\qquad j,k\in\A,
\end{align*}
where $p_x$ was defined in Section~\ref{s:TSE} as functions of $\Gm_{x}(z,\Lp)$. Therefore, the matrix $P$ depends continuously on $z$.
In the case of regular trees there is only one label, so the matrix $P$ is then a number. In this sense, the proposition below can be considered as a higher dimensional analogue of \cite[Theorem~6]{FHS2}.

Denote  $V=\bigcup_{j\in \A} V_{j}$. Then, $\ell^{2}(V)=\bigoplus_{j\in \A}\ell^{2}(V_{j})$.
Let now random variables $(v,\te):V\to(-1,1)\times(-1,1)$ satisfying $\mathrm{(P1)}$ and $\mathrm{(P2)}$ be given. We get random operators $H^{\lm,\om}$ on $\ell^{2}(V)$ which decompose into a direct sum of operators on $\ell^{2}(V_{j})$.
For $p>1$, $z\in\h$ and $\lm\geq0$ denote
\begin{align*}
\EE\gm:=\left(\EE\left( \gm(\Gm_{o(j)}(z,H^{\lm}),\Gm_{o(j)}(z,\Lp))^{p} \right)\right)_{j\in\A}.
\end{align*}

The goal of this section is to prove the following vector inequality. Its proof relies on the two step expansion estimate, Proposition~\ref{p:expansion}, and the uniform contraction estimate, Proposition~\ref{p:ka}.

\begin{prop}\label{p:VI}(Vector inequality) Let $I\subset \Sigma$ be  compact and $p>1$. Then there are $\de= \de(I, p) > 0$, $\lm_{0}=\lm(I,p)>0$ and $C : [0,\lm_0) \to [0,1)$ with $\lim_{\lm\to 0}C(\lm)= 0 $ such that
\begin{align*}
    \EE\gm\leq(1-\de) P\EE\gm + C(\lm)
\end{align*}
for all $z\in I+i[0,1]$, $\lm\in[0,\lm_0)$ and all $(v,\te)$ satisfying $\mathrm{(P1)}$ and $\mathrm{(P2)}$.
\end{prop}
\begin{proof}
Let $I\subset \Sigma$ be compact and $p>1$. Denote $o=o(j)$ for $j\in\A$.  We let $z\in I+i[0,1]$ and $g\in\h^{S_{o,o'}}$ be given by the random variables   $g_x^{\om}=(1+\lm\te_{x}^{\om})\Gm_{S_{x}}(z,H^{\lm,\om})$, $x\in S_{o,o'}$.
By \eqref{e:rec} and the definition of $g_{o'}$, we have $g_{o'}(z,\lm v_{o}^{\om},g)=\Gm_{o'}(z,H^{\lm,\om})$. Moreover, by $\mathrm{(P1)}$ and $\mathrm{(P2)}$, the random variables $g_x$ and $g_{y}$ are identically distributed for all vertices $x$, $y$ that carry the same label and independent for all $x,y\in S_{o,o'}$. This gives, in particular, $\EE(\sum_{x\in S_{o,o'}} p_{x}c_{x}(g)\gm_{x}(g)^p)=\EE\big(\sum_{x\in S_{o,o'}} p_{x}c_{x}^{(\pi)}(g)\gm_{x}^{(\pi)}(g)^p\big)$ for all $\pi\in\Pi$.
We use this to compute
\begin{align*}
\EE\Big(\Big(\sum_{x\in S_{o,o'}}p_{x}c_{x}(g)\gm_{x}(g)\Big)^p\Big) &=\frac{1}{|\Pi|}\EE\Big(\sum_{\pi\in\Pi}\Big(\sum_{x\in S_{o,o'}}p_{x}c_{x}^{(\pi)}(g)\gm_{x}^{(\pi)}(g)\Big)^p\Big) \\
& =\frac{1}{|\Pi|}\EE\Big({\ka_{o}^{(p)}(z,\lm v_{o'}^{\om},g){\sum_{\pi\in\Pi}\sum_{x\in S_{o,o'}}p_{x}\gm_{x}^{(\pi)}(g)^p}}\Big) .
\end{align*}
In order to apply the uniform contraction estimate, Proposition~\ref{p:ka}, we split up the expectation value. Let $\lm_0':=\max_{k\in\A}\lm_{o(k)}(I)>0$ and $R:=\max_{k\in\A}R_{o(k)}:[0,\lm_0']\to[0,\infty)$ where $\lm_{o(k)}$ and $R_{o(k)}$, $k\in\A$, are given by Proposition~\ref{p:ka}. Hence, $\lim_{\lm\to0}R(\lm)=0$.
For $\lm\in[0,\lm_0']$ let $\mathds{1}_{R}$ be the characteristic function of $B_{R(\lm)}=\{g\in \h^{ S_{o,o'}}\mid \gm(g_x,\Gm_{x}(z,\Lp))\leq R(\lm)\mbox{ for all } x\in S_{o,o'}\}$ and $\mathds{1}_{R}^{c}$ be the characteristic function of its complement $B_{R(\lm)}^{c}=\h^{S_{o,o'}}\setminus B_{R(\lm)}$. We proceed   using  $\ka_o^{(p)}=\ka_{o}^{(p)}(z,\lm v_{o'}^{\om},g)\leq1$ in the second term
\begin{align*}
\ldots&\leq\frac{1}{|\Pi|}\EE\Big({\ka_{o}^{(p)} \sum_{\pi\in\Pi}\sum_{x\in S_{o,o'}}p_{x}{\gm_{x}^{(\pi)}(g)}^p\mathds{1}_{R}^{c}(g)} \Big)+\frac{1}{|\Pi|} \EE\Big({\sum_{\pi\in\Pi}{\sum_{x\in S_{o,o'}}p_{x}{\gm_{x}^{(\pi)}(g)}^p} \mathds{1}_{R}(g)}\Big),\\
&=\frac{1}{|\Pi|}\EE\Big({\ka_{o}^{(p)} \sum_{\pi\in\Pi}{\sum_{x\in S_{o,o'}}p_{x}{\gm_{x}^{(\pi)}(g)}^p}\mathds{1}_{R}^{c}(g)} \Big) +\EE\Big({{\sum_{x\in S_{o,o'}}p_{x}{\gm_{x}^{(\pi)}(g)}^p}\mathds{1}_{R}(g)}\Big),
\end{align*}
where we used $\EE\ap{\sum_\pi {\sum_{x\in S_{o,o'}}p_{x}\gm_{x}^{(\pi)}(g)^p}\mathds{1}_{R}(g)} =|\Pi|\EE\ap{{\sum_{x\in S_{o,o'}}p_{x}\gm_{x}(g)^p}\mathds{1}_{R}(g)}$ for the second term, as $B_{R(\lm)}$ is $\pi$ invariant. We now apply the uniform contraction estimate, to the first term with $\de_0:=\max_{k\in\A}\de_{o(k)}(I,p)$ and $\de_{o(k)}$  taken from  Proposition~\ref{p:ka}. For the second term, note that ${\sum_{x}p_{x}\gm_{x}(g)^p}\leq \max_{x}\gm(g_{x},\Gm_{x}(z,\Lp))^{p}\leq R(\lm)^{p}$ for $g\in B_{R(\lm)}$ by Lemma~\ref{l:Z}. This gives
\begin{align*}
\ldots&\leq (1-\de_0)\frac{1}{|\Pi|}\EE\Big({ \sum_{\pi\in\Pi}{\sum_{x\in S_{o,o'}} p_{x}\gm_{x}^{(\pi)}(g)^p} \mathds{1}_{R}^{c}(g)}\Big) +R(\lm)^p\\
&\leq(1-\de_0){{\sum_{x\in S_{o,o'}}p_{x}\EE\big(\gm_{x}(g)^p\big)}}+R(\lm)^p,
\end{align*}
where we used the $\pi$ invariance of $B_{R(\lm)}^{c}$ and $\EE({\gm_{x}(g)^{p}\mathds{1}_{R}^{c}})\leq \EE({\gm_{x}(g)^{p}})$ in the second step.
In summary, this yields
\begin{align}\label{e:EGm}\tag{$\spadesuit$}
\EE\Big(\Big({\sum_{x\in S_{o,o'}}p_{x}c_{x}(g)\gm_{x}(g)}\Big)^p\Big)  \leq (1-\de_0) \sum_{x\in S_{o,o'}}p_{x}\EE\ap{\gm_{x}(g)^{p}}+R(\lm)^p.
\end{align}
We want to combine this with the two step expansion estimate, Proposition~\ref{p:expansion}. Let us first observe the  following consequence of
Jensen's inequality
\begin{align}\label{e:jensen}\tag{$\heartsuit$}
(r+s)^{p}=\Big({\frac{1}{1+s}(1+s)r+\frac{s}{1+s}(1+s)}\Big)^{p} \leq (1+s)^{p-1}r^{p}+(1+s)^{p}s,
\end{align}
for $ r,s\geq0$.
We denote $\Gm_{x}^{\lm,\om}=\Gm_{x}(z,H^{\lm,\om})$, $\om\in\Om$, and $\Gm_{x}=\Gm_{x}(z,\Lp)$.
For each $j\in\A$, we apply Proposition~\ref{p:expansion}, inequality \eqref{e:jensen} and the inequality above to obtain
\begin{align*}
\EE\Big(\gm(\Gm_{o(j)}^{\lm},\Gm_{o(j)})^p\Big) &\leq\EE\Big(\Big((1+c(\lm)){\sum_{x\in S_{o,o'}}p_{x}c_{x}(g)\gm_{x}(g)} +c(\lm)\Big)^p\Big)\\
&\leq(1+c(\lm))^{2p-1}\EE\Big(\Big(\sum_{x\in S_{o,o'}}p_{x}c_{x}(g)\gm_{x}(g)\Big)^p\Big) +(1+c(\lm))^{p-1}c(\lm)
\end{align*}
Recall that $g_{x}^{\om}=(1+\lm \te_{x}^{\om})\Gm_{x}(z,H^{\lm,\om})$, $\om\in\Om$.
We now apply Lemma~\ref{l:ti}  to $\gm_{x}(g)$ with $a=(1+\lm \te_{x}^{\om})$ and $b=0$
to get with inequality \eqref{e:jensen}
\begin{align*}
\gm_{x}(g)^{p} \leq\ab{(1+c_0'(\lm))\gm(\Gm_{x}^{\lm},\Gm_{x})+c_{0}'(\lm)}^{p} \leq(1+c_{0}'(\lm))^{2p-1} \gm(\Gm_{x}^{\lm},\Gm_{x})^{p}+(1+c_{0}'(\lm))^{p}c_{0}'(\lm).
\end{align*}
Since $\Gm_x^{\lm}$ is equal to $\Gm_{o(k)}^{\lm}$ in distribution whenever $x$ carries label $k$, we get  by the definition  of $P$ and the estimate \eqref{e:EGm} that
\begin{align*}
\EE\ap{\gm(\Gm_{o(j)}^{\lm},\Gm_{o(j)})^p}
\leq(1+c_{1}(\lm))(1-\de_0) \sum_{k\in\A}P_{j,k}\EE\ap{\gm_{o(k)}(\Gm_{o(k)}^{\lm})^{p}}+C(\lm).
\end{align*}
where $(1+c_1(\lm))=(1+c(\lm))^{2p-1}(1+c_0'(\lm))^{2p-1}$
and $C(\lm)=(1+c(\lm))^{p-1}c(\lm)+(1+c'_{0}(\lm))^{p-1}c_{0}'(\lm) +R(\lm)$ satisfy $c_{1}(\lm), C(\lm)\to 0$ since  $c(\lm),c_{0}'(\lm),R(\lm)\to0$ for $\lm\to0$  by Proposition~\ref{p:expansion}, Lemma~\ref{l:ti} and Proposition~\ref{p:ka}. As $c_{1}(\lm)\to0$ for $\lm\to0$, there is $\lm_0>0$ and $\de>0$ that satisfy the assertion of the proposition. Thus, we finished the proof.
\end{proof}

\section{Proof of the theorems}\label{s:proofs}
Let $\T$ be a rooted tree with root  $o$ and $x_0$ be an arbitrary vertex. Considering $x_0$ as the new root of $\T$, we denote the truncated Green functions of a self adjoint operator $H$ with respect to the root $x_0$ by $\Gm_{x}^{(x_0)}(z,H)$. With this notation we have $\Gm_{x}(z,H)=\Gm^{(o)}_{x}(z,H)$.

\begin{lemma}\label{l:G} Let $x_0$ be a vertex in a rooted tree $\T$ with root $o$ and $\Sigma_{0}\subset\R$ be such that $\h\to\h$, $z\mapsto\Gm_{x}(z,H)=\Gm_{x}^{(o)}(z,H)$ extends to a continuous function $\h\cup\Sigma_{0}\to\h$. Then $\h\to\h$, $z\mapsto\Gm_{x}^{(x_{0})}(z,H)$ extends to a continuous function $\h\cup \Sigma_{0}\to\h$.
\end{lemma}
\begin{proof} If $N$ is the distance of $x_0$ to the root $o$, then we have for all vertices $x$ of distance $N+1$ from $x_0$ that $\Gm_{x}^{(x_0)}(z,H)=\Gm_{x}(z,H)$. Hence, they have continuous extensions from $\h\cup \Sigma_{0} $ to $\h$. Suppose we have shown the statement for all vertices of distance $n\leq N$ from $x_0$. Let $x$ be a vertex with distance $n-1$ from $x_0$. If $\Gm^{(x_0)}_{x}(E+i\eta,H)$ converges neither to $0$ nor to $\infty$ as $\eta\to0$, then continuity follows from \eqref{e:rec}.
Taking the modulus  in \eqref{e:rec} gives
\begin{align*}\frac{1}{|\Gm_{x}^{(x_0)}(z,H)|}\geq {\sum_{y }|t(x,y)|^{2}\Im \Gm_{y}^{(x_0)}(z,H)} \end{align*}
where the sum is over all forward neighbors of $x$ with respect to $x_0$. Since $\Im\Gm_{y}^{(x_0)}(z,H)$ are assumed to stay positive in the limit $\eta\to0$, it follows that $\Gm_{x}^{(x_0)}$ stays bounded. On the other hand, taking imaginary parts in \eqref{e:rec} and multiplying by $|\Gm_{x}^{(x_0)}(z,H)|^{2}$ yields
\begin{align*}     \Im \Gm_{x}^{(x_0)}(z,H) \geq{\sum_{y}|t(x,y)|^{2}\Im \Gm_{y}^{(x_0)}(z,H)}|\Gm_{x}^{(x_0)}(z,H)|^{2}, \end{align*}
If $\Im \Gm_{x}^{(x_0)}(z,H)$ is zero, then $|\Gm_{x}^{(x_0)}(z,H)|$ must be zero as $\Im \Gm_{y}^{(x_0)}(z,H)>0$ by assumption. However, this is impossible as \eqref{e:rec} then yields that at least one of the  $ \Gm_{y}^{(x_0)}(z,H)$ tends to $\infty$ as $\eta\to0$.
\end{proof}

Before we come to the proof of Theorem~\ref{main1}, we deduce from the vector inequality that all moments of the $\gm$-distances  of the Green functions are bounded.

\begin{prop}\label{p:EEG}Let $\T$ be a  rooted tree such that the forward trees of all vertices from a certain sphere on are generated by substitution matrices that satisfy $\mathrm{(M0)}$, $\mathrm{(M1^*)}$ and $\mathrm{(M2)}$.
For all  $I\subset\Sigma$ compact and  $p>1$, there is $\lm_0=\lm_0(I,p)>0$ and $C_{x}:[0,\lm_0)\to[0,\infty)$ for $x\in V$ with $\lim_{\lm\to0}C(\lm)=0$ such that
\begin{align*}
\sup_{z\in I+i(0,1]} \EE\ab{\gm(G_{x}(z,H^{\lm}),G_{x}(z,\Lp))^{p}} \leq C(\lm)
\end{align*}
 for all $\lm\in[0,\lm_0)$ and  all $(v,\te)$ satisfying $\mathrm{(P1)}$ and $\mathrm{(P2)}$.
\end{prop}
\begin{proof}
Let $N$ be the sphere from which on all forward trees are generated by substitution matrices satisfying $\mathrm{(M0)}$, $\mathrm{(M1^*)}$ and $\mathrm{(M2)}$. Let $\T_{x}$ be such a forward tree generated by a substitution matrix $M$ over a label set $\A$. Let $P$ be the stochastic matrix defined before the vector inequality Proposition~\ref{p:VI}. By $\mathrm{(M1^*)}$ the entries of $P$ are positive whenever the entries of $M$ are. Therefore, by $\mathrm{(M2)}$ the matrix $P$ is irreducible and by the Perron Frobenius theorem there is a positive left eigenvector $u\in\R^{\A}$, such that $P^{\top}u=u$. Hence, letting $I\subset\Sigma$ be compact and $p>1$, we obtain by the vector inequality, Proposition~\ref{p:VI},
\begin{align*}
   \as{u, \EE\gm} \leq(1-\de) \as{u,P\EE\gm} + C(\lm)=(1-\de) \as{u,\EE\gm} + C(\lm)
\end{align*}
for all $\lm\in[0,\lm_0)$ and all $z\in I+i[0,1]$ and $\as{\cdot,\cdot}$ denotes the standard scalar product in $\R^{\A}$. This implies $\EE(\gm_{o(j)}^{p})\leq {C(\lm)}/({u_{j}\de})$, $j\in\A'$ and $u_{j}$ is uniformly bounded away from zero since $I+i[0,1]$ is compact. Thus, for all vertices $x$ that have distance larger or equal to $N$ from the root, it follows
\begin{align}\label{e:EE}\tag{$\diamondsuit$}
\EE(\gm(\Gm_{x}(z,H^{\lm})\Gm_{x}(z,\Lp))^{p})\leq {C_0(\lm)}
\end{align}
with $\lim_{\lm\to0}C_0(\lm)=0$ which can be chosen independently of $x$.

Let $x_0$ be an arbitrary vertex in the tree $\T$ with distance $N_0$ from the root. We reorder the tree $\T$ with respect to $x_0$ as a root.  We figure that the forward trees of all vertex in the $(N+N_0)$ sphere $S_{N+N_0}(x_0)$ of $x_0$ are generated by substitution matrices satisfying $\mathrm{(M0)}$, $\mathrm{(M1^*)}$, $\mathrm{(M2)}$ as $\T_{x}$ above. By \eqref{e:rec} we can expand $G_{x}(z,H^{\lm,\om})$ and $G_{x}(z,\Lp)$  in terms of its forward neighbors in $\T$ and note that after $(N+N_0)$ steps the Green functions of all forward neighbors are such that they satisfy \eqref{e:EE}. Applying the one step expansion estimate, Lemma~\ref{l:OSE}, and the linear perturbation estimate, Lemma~\ref{l:ti}, $(N+N_0)$ times, we obtain, by  estimating all terms $\Im \Gm_{x_i}^{(x_0)}/\sum \Im \Gm_{y}^{(x_0)}\leq 1$ and
$\sum_{y\in S_{x_i}}q_{y}Q_{x,y}\cos\al_{x,y}\leq 1$ ,
\begin{align*}
\gm(G_{x_{0}}(z,H^{\lm,\om}),G_{x_{0}}(z,\Lp))\leq (1+ c(\lm))\sum_{x\in S_{N+N_0}(x_0)} \gm(\Gm_{x}(z,H^{\lm,\om}),\Gm_{x}(z,\Lp))+ c(\lm).
\end{align*}
Here, $c(\lm)$ is a finite sum of products of the $c_{0}(\lm,h)$ from the linear perturbation estimate, Lemma~\ref{l:ti}, which can be chosen independently of $z$ by Lemma~\ref{l:G} above and, furthermore, such that $\lim_{\lm\to0}c(\lm)=0$.
Combining this with \eqref{e:jensen},  we get
\begin{align*}
{\EE(\gm(G_{x_{0}}(z,H^{\lm}),G_{x_{0}}(z,\Lp))^{p})}
\leq c'(\lm)\sum_{x\in S_{N+N_0}(x_0)} \EE(\gm(\Gm_{x}(z,H^{\lm}),\Gm_{x}(z,\Lp))^{p})+ c''(\lm),
\end{align*}
with $c'(\lm)=2^{|S_{N+N_0}(x_0)|}(1+ c(\lm))^{2p-1}$ and $c''(\lm)=(1+c(\lm))^{p-1}c(\lm)$. By \eqref{e:EE} the statement follows.
\end{proof}

\begin{proof}[Proof of Theorem~\ref{main3}]
Put $g_{x}^{\om}=G_{x}(z,H^{\lm,\om})$, $h_{x}=G_{x}(z,\Lp)$ and let $I\subset\Sigma$ be compact.
By the Cauchy-Schwarz inequality and Proposition~\ref{p:EEG} we get
\begin{align*}
\EE\ap{|g_{x}-h_{x}|^{p}}^2\leq\EE\ap{\gm(g_{x},h_{x})^{p}} \EE\ap{\ap{\Im{g_{x}}\Im{h_{x}}}^{p}}\leq C(\lm)\EE\ap{\mo{g_{x}}^{p}}|h_{x}|^{p}
\end{align*}
for all $\lm\in[0,\lm_{0})$.
We now  take the supremum over all $z\in I+i[0,1]$. Thus, it remains to check that $\sup_{z\in I+i[0,1]}\EE\ap{\mo{g_{x}}^{p}}|h_{x}|^{p}< \infty$.
Firstly, by Proposition~3 of \cite{KLW}, the Green functions $G_x$ are uniformly bounded, whenever the truncated Green functions $\Gm_{x}$ are.
Therefore, we conclude $\sup_{z\in I+i[0,1]}|h_{x}|<\infty$. Secondly,
we employ the inequality
\begin{align*}
\mo{\xi}\leq 4\gm(\xi,\zeta)\Im \zeta+2|\zeta|,\qquad \xi,\zeta\in\h,
\end{align*}
from \cite{FHS2}. (This inequality is obvious for $|\xi|\leq 2| \zeta|$ and follows from
$|\xi|\Im \xi\leq|\xi|^{2}\leq 2|\xi-\zeta|^2+2|\zeta|^2\leq 4|\xi-\zeta|^2$ for $|\xi|\geq2|\zeta|$.)
This inequality applied with $\xi=g_{x}$ and $\zeta=h_{x}$ combined with Proposition~\ref{p:EEG} and Jensen's inequality gives
\begin{align*}
\sup_{z\in I+i(0,1]} \EE\ap{|g_{x}|^{p}}<\infty.
\end{align*}
This finishes the proof of the first statement of Theorem~\ref{main3}.

For the second statement, we write $z=E+i\eta$ and denote by $\Leb(I)$ the Lebesgue measure of $I$.
The previous inequality also implies by Fatou's lemma and Fubini's theorem that
\begin{align*}
\EE\ap{\liminf_{\eta\to 0}\int_{I}|g_{x}|^{p}dE}\leq \liminf_{\eta\to0}\int_{I}\EE\ap{|g_{x}|^{p}}dE
\leq\sup_{z\in{I+i[0,1)}}\EE\ap{|g_{x}|^{p}} \Leb(I) <\infty.
\end{align*}
for all $\lm\in[0,\lm_{0})$. Hence, we have
\begin{align*}
{\liminf_{\eta\downarrow0}\int_{I}|g_{x}^{\om}|^{p}dE}<\infty.
\end{align*}
almost surely. This yields the absence of singular spectrum by a theorem of Klein \cite[Theorem~4.1]{Kl1}. Moreover, $\Im G_{x}(E+i\eta,H^{\lm,\om})$ can only tend to zero as $\eta\to0$ on subsets of $I\times\Om$ of $(\Leb\otimes\PP)$-measure zero. Otherwise, this would lead to a contradiction to Proposition~\ref{p:EEG}. Hence, we conclude  $I\subseteq \si(H^{\lm,\om})$ almost surely.
\end{proof}

\begin{proof}[Proof of Theorem~\ref{main1}] By Proposition~\ref{p:T}  the set $\Sigma$ consists of finitely many intervals and $\clos\Sigma=\si(\Lp)$. Hence, there is a finite set $\Sigma_0$ such that $\Sigma=\si(\Lp)\setminus\Sigma_0$. Therefore, the statement follows from Theorem~\ref{main3}.
\end{proof}

\textbf{Acknowledgements.} This work was started while M.K. was visiting Princeton University and finished during his visit at the Hebrew University where he was supported by the Israel Science Foundation (Grant no. 1105/10). He would like to thank the Department of Mathematics in both places for their hospitality.

\end{document}